\documentclass[a4paper,12pt,liststotoc,idxtotoc,bibtotoc,normalheadings,BCOR16mm,DIV12,headsepline=false,headinclude,twoside,cleardoublepage=empty]{scrreprt}

\usepackage[pdftex,bookmarks]{hyperref}
\usepackage{graphicx}
\usepackage{amsmath}
\usepackage{amssymb}
\usepackage{wasysym}
\usepackage[utf8]{inputenc}
\usepackage[T1]{fontenc}
\usepackage[pdftex,usenames]{color}
\usepackage{subfig} 
\usepackage[final]{pdfpages}

\usepackage[english]{babel}
\usepackage{dsfont}
\usepackage{bbm}
\usepackage{amsmath, amsthm, amssymb, mathrsfs}
\usepackage[pdftex,usenames]{color}
\usepackage{setspace}

\newtheorem{lemma}{Lemma}
\newtheorem{theorem}{Theorem}

\numberwithin{equation}{section}
\numberwithin{theorem}{section}
\numberwithin{lemma}{section}
\setkomafont{descriptionlabel}{\bf}
\DeclareMathOperator{\Tr}{\mathrm{Tr}}%trace2
\DeclareMathOperator{\probability}{\mathrm{Pr}}%Pr
\DeclareMathOperator{\tracedistance}{\mathcal{D}}
\DeclareMathOperator{\ee}{\mathrm{e}}%Euler e
\DeclareMathOperator{\iu}{\mathbbm{i}}%imaginary unit
\renewcommand\Re{\operatorname{Re}}
\renewcommand\Im{\operatorname{Im}}
\DeclareMathOperator{\hiH}{\mathcal{H}}%Hilbert spaces
\DeclareMathOperator{\haH}{\mathscr{H}}%Hamiltonians

\newcommand{\bra}[1]{\langle #1|}
\newcommand{\ket}[1]{|#1\rangle}
\newcommand{\braket}[2]{\langle #1|#2\rangle}
\newcommand{\ketbra}[2]{| #1 \rangle \langle #2 |}
\newcommand{\expect}[1]{\left\langle #1\right\rangle}

\definecolor{structure}{rgb}{0.23,0.4,0.7}

\newcommand{\inputTikZ}[1]{\includegraphics{#1}}

\usepackage{avant}

\setkomafont{pagehead}{\normalcolor\sffamily\small}
\setkomafont{pagenumber}{\normalcolor\sffamily\small}

\usepackage[automark]{scrpage2}
\makeindex

\begin{document}

\pagenumbering{roman}

% %%%% Title %%%%%%%%%%%%%%%%%%%%%%%%%%%%%%%%%%%%%%%%%%%%%
% \titlehead{\begin{center}\vspace{-1.4cm}\vbox{\sffamily{\Large Julius-Maximilians-Universität Würzburg}\\Institut für theoretische Physik und Astronomie}\end{center}}
% \title{Pure State Quantum Statistical Mechanics}
% \author{Christian Gogolin}
% \subject{\normalfont\normalsize\sffamily Arbeit zum Erwerb des akademischen Grades\\Master of Science}
% \date{\today}
% \publishers{\normalsize
%   \begin{tabbing}%
%     Abgabetermin\hspace{0,5cm} \= \kill
%     Vorgelegt von: \> Christian Gogolin, \\
%                    \> geboren am 16.08.1985 in Karlsruhe \\
%     Betreuer:      \> Prof. Dr. Haye Hinrichsen\\
%     Zweitgutachter:\> Prof. Dr. Andreas Winter\\
%     Lehrstuhl:     \> Theoretische Physik III\\
% %    Abgabetermin:  \> \rule{35mm}{0.1pt}\\
%   \end{tabbing}
%   \begin{tikzpicture}[remember picture,overlay]
%     \pgfdeclareimage[width=33cm,height=33cm]{logo}{neuSIEGEL}
%     \node[opacity=0.05,left=2cm,below=-14cm] at (current page.center) {
%       \begin{pgfpicture}
%         \pgftext{\pgfuseimage{logo}}
%       \end{pgfpicture}
%     };
%   \end{tikzpicture}
% }

%%%% Title for arXiv %%%%%%%%%%%%%%%%%%%%%%%%%%%%%%%%%%%
\titlehead{\begin{center}\vspace{-1.4cm}\vbox{\sffamily{\Large Julius-Maximilians-Universität Würzburg}\\Institut für theoretische Physik und Astronomie}\end{center}}
\title{Pure State Quantum Statistical Mechanics}
\author{Christian Gogolin}
\subject{\normalfont\normalsize\sffamily Master's Thesis in Theoretical Physics}
\date{\today}
\publishers{\normalsize
  \begin{tabbing}%
    Author: \hspace{1,5cm} \= Christian Gogolin\footnote{\href{mailto:publications@cgogolin.de}{publications@cgogolin.de}}, \\
    Supervisors:           \> Prof. Dr. Haye Hinrichsen\\
                           \> Prof. Dr. Andreas Winter\\
    Institute:             \> Julius-Maximilians-Universität Würzburg\\
                           \> Theoretische Physik III\\
  \end{tabbing}
%   \begin{tikzpicture}[remember picture,overlay]
%     \pgfdeclareimage[width=33cm,height=33cm]{logo}{neuSIEGEL}
%     \node[opacity=0.05,left=2cm,below=-14cm] at (current page.center) {
%       \begin{pgfpicture}
%         \pgftext{\pgfuseimage{logo}}
%       \end{pgfpicture}
%     };
%   \end{tikzpicture}
}

{\sffamily \maketitle}
\thispagestyle{empty}
\begin{abstract}
  \section*{Abstract}
  The capabilities of a new approach towards the foundations of Statistical Mechanics are explored.
  The approach is genuine quantum in the sense that statistical behavior is a consequence of objective quantum uncertainties due to entanglement and uncertainty relations.
  No additional randomness is added by hand and no assumptions about a priori probabilities are made, instead measure concentration results are used to justify the methods of Statistical Physics.
  The approach explains the applicability of the microcanonical and canonical ensemble and the tendency to equilibrate in a natural way.

  This work contains a pedagogical review of the existing literature and some new results.
  The most important of which are:
  i)~A measure theoretic justification for the microcanonical ensemble.
  ii)~Bounds on the subsystem equilibration time.
  iii)~A proof that a generic weak interaction causes decoherence in the energy eigenbasis. 
  iv)~A proof of a quantum H-Theorem.
  v)~New estimates of the average effective dimension for initial product states and states from the mean energy ensemble.
  vi)~A proof that time and ensemble averages of observables are typically close to each other.
  vii)~A bound on the fluctuations of the purity of a system coupled to a bath.
\end{abstract}

\cleardoublepage

%%%%Dedication%%%%%%%%%%%%%%%%%%%%%%%%%%%%%%%%%%%%%%%%%
\thispagestyle{empty}
\vspace*{6cm}
\begin{center}
  This work is dedicated to Kathrin and Meggy,\\ the two most important persons in my life.
\end{center}
\cleardoublepage

%%%%Quotation%%%%%%%%%%%%%%%%%%%%%%%%%%%%%%%%%%%%%%%%%%
\thispagestyle{empty}
\vspace*{6cm}
\begin{quotation}
  A philosopher once said ``It is necessary for the very existence of
  science that the same conditions always produce the same results.''
  Well, they do not.
\end{quotation}
\begin{flushright}
  Richard Feynman, \emph{The Character of Physical Law}  
\end{flushright}
\cleardoublepage

%%%%Table of contents%%%%%%%%%%%%%%%%%%%%%%%%%%%%%%%%%%
\tableofcontents
%\listoftodos
\cleardoublepage

%%%%Notation guide%%%%%%%%%%%%%%%%%%%%%%%%%%%%%%%%%%%%%
\manualmark
\markright{Notation guide and definitions}
\addcontentsline{toc}{chapter}{Notation guide and definitions}
\chapter*{Notation guide and definitions}
\label{sec:notationguideandsomedefinitions}
\begin{description}
\item[Hilbert spaces] $\hiH\ \hiH_S\ \hiH_B\ \hiH_R, \dots$
\item[Hamiltonians] $\haH,\ \haH_S,\ \haH_B,\ \haH_{SB}, \dots$
  \begin{align*}
    \text{eigenvectors } &\ket{E_k},\ \ket{E_l},\ \ket{E_m}, \dots \\
    \text{eigenvalues } & E_k,\ E_l,\ E_m, \dots
  \end{align*}
\item[observables and projectors]
  \begin{align*}
    \text{observables } &A,\ B, \dots \\
    \text{projectors } &\Pi \\
    \text{rank } n \text{ projectors } &\mathcal{P}_n(\hiH) \\
    \text{all projectors } &\mathcal{P}(\hiH)
  \end{align*}
\item[quantum states]
  \begin{align*}
    \text{pure states } &\psi,\ \varphi \in \mathcal{P}_1(\hiH) \\
    \text{mixed states } &\rho,\ \sigma \in \mathcal{M}(\hiH) \\
    \text{reduced states/marginals } &\rho^S = \Tr_B[\rho] \in \mathcal{M}(\hiH_S),\ \rho^B = \Tr_S[\rho] \in \mathcal{M}(\hiH_B) \\
    \text{time averaged/dephased states } &\omega = \expect{\rho_t}_t = \$[\rho_0] := \sum_k \ketbra{E_k}{E_k} \rho_0 \ketbra{E_k}{E_k}
  \end{align*}
\item[trace norm]
  \begin{equation}
    \|\rho\|_1 = \Tr|\rho| = \Tr[\sqrt{\rho^\dagger\,\rho}]
  \end{equation}
\item[trace distance]
\begin{align}
  \label{eq:definitiontracedistance}
  \tracedistance(\rho,\sigma) &= \frac{1}{2} \|\rho -\sigma \|_1= \frac{1}{2} \Tr |\rho - \sigma| \\
  &= \max_{0\leq A \leq \mathds{1}} \Tr[A (\rho - \sigma)] \\
  &= \max_{\Pi  \in \mathcal{P}(\hiH)} \Tr[\Pi  (\rho - \sigma)] 
\end{align}
\item[Hilbert space norm]
\begin{equation}
  \| \ket{\psi} \|_2 = \sqrt{\braket{\psi}{\psi}} = \sqrt{\|\psi\|_1} 
\end{equation}
\item[Hilbert-Schmidt norm]
\begin{equation}
  \|\rho\|_2 = \sqrt{\Tr[A^\dagger\,A]}
\end{equation}
\item[operator norm of a hermitian operator $A$]
\begin{equation}
  \|A\|_\infty = \max_{\psi \in \mathcal{P}_1(\hiH)} Tr[A\,\psi] 
\end{equation}
\item[Von~Neumann entropy]
\begin{equation}
  S(\rho) = - \Tr[\rho\,\log(\rho)] ,
\end{equation}
\item[quantum mutual information between $S$ and $B$]
\begin{equation}
  \label{eq:mutualinformation}
  \begin{split}
    I_{SB}(\rho_t) &= S(\rho^S_t) + S(\rho^B_t) - S(\rho_t)\\
    &= \Tr[\rho_t \log(\rho_t) - \rho_t \log(\rho^S_t \otimes \rho^B_t)]
  \end{split}
\end{equation}
\item[purity]
\begin{equation}
  p(\rho) = \Tr[\rho^2]
\end{equation}
\item[effective dimension]
\begin{equation}
  d^{\mathrm{eff}}(\omega) = \frac{1}{\Tr[\omega^2]}
\end{equation}
\end{description}
\newpage
\pagenumbering{arabic}

\automark[chapter]{section}
\chapter{Introduction}
\label{sec:introduction}
Despite being very well confirmed by experiments Thermodynamics and classical Statistical Physics still lack a commonly accepted and conceptually clear foundation.

The reason for this unsatisfactory situation is that physicists have not yet succeeded in finding concise and convincing justifications for the fundamental axioms of Statistical Physics.
An overview of the attempts to axiomatize Statistical Physics and Thermodynamics and to justify the axioms from classical Newtonian Mechanics and the conceptual problems with these approaches can be found for example in \cite{UffinkFinal} and \cite{RevModPhys.27.289} and the references therein.

Quantum Mechanics claims to be a fundamental theory.
As such it should be capable of providing us with a microscopic explanation for all phenomena we observe in macroscopic systems, including irreversible processes like thermalization.
But, its unitary time evolution seems to be incompatible with irreversibility \cite{FeynmanV01} leading to an apparent contradiction between Quantum Mechanics and Thermodynamics.
This apparent contradiction is part of the long standing problem of the emergence of classically from Quantum Mechanics.

To overcome this problem many authors have suggested to modify Quantum Theory, either by adding nonlinear terms to the von~Neumann equation or by postulating a periodical spontaneous collapse of the wave function \cite{Bassi03}.
Others have considered effective, Markovian, time evolutions for open quantum systems \cite{Breuer02} and it has been shown that system bath models that evolve under a special form of Hamiltonian tend to evolve into states that are classical superpositions of so called \emph{pointer states} --- a phenomenon called \emph{environmentally induced super selection}, a term due to Zurek \cite{RevModPhys.75.715}.
Depending on the author subsets of these approaches are subsumed under the term \emph{decoherence theory} \cite{zeh96,Breuer02,Gemmer09,Hornberger09}.

In face of the enormous success of standard Quantum Mechanics in explaining microscopic phenomena and the additional difficulties that arise when the von~Neumann equation is modified and the existence of macroscopic quantum systems on the one hand, and the broad applicability of Statistical Mechanics and Thermodynamics on the other, we feel that neither a modification of Quantum Theory, nor considerations restricted to special situations can provide a satisfactory explanation of the statistical and thermodynamic behavior of our macroscopic world.
Consequently we will seek to derive general statements independent of particular models and we will not use the Markov assumption.
Furthermore, we believe that neither the assumption of ergodicity nor classical or quantum chaos are good starting points for constructing a convincing and consistent foundation for Statistical Mechanics and Thermodynamics (see for example footnote 1 and 2 in \cite{tasaki98}).

The struggle for a quantum mechanical explanation of behavior usually described by Statistical Physics dates back to the founding fathers of Quantum Theory, most notably von~Neumann \cite{vonneumann1929} and Schrödinger \cite{Schroedinger27}.
Recently work on this subject was resumed and there has been remarkable success:
\begin{itemize}
\item In \cite{slloydthesis,tasaki98,Gemmer02,Popescu05,Popescu06,Goldstein06,Cho09} a justification for the applicability of the canonical ensemble is given that does not rely on subjective, added randomness or ensemble averages.
While \cite{tasaki98,Gemmer02,Goldstein06} make particular assumptions on the Hamiltonian and introduce the concept of temperature, and thereby are able to derive explicitly the Boltzmann distribution, the aim of \cite{slloydthesis,Popescu05,Popescu06} is more to show that the reduced states of random states of large quantum systems typically look like the reduced state of the microcanonical state, \cite{Cho09} in addition uses time dependent perturbation theory. 
All these works are based on typicality arguments and the phenomenon of measure concentration \cite{ledoux01}.\footnote{It is very interesting to compare thees articles with the works of Jaynes \cite{PhysRev.106.62,PhysRev.108.17}
Although there are huge differences concerning the interpretation, the before mentioned works are methodologically very close to certain aspects of the approach of Jaynes, especially with respect to the way they make use of measure concentration arguments. It is thus surprising and unfortunate that Jaynes' works have been completely ignored in the recent literature.}.
\item In \cite{PhysRevE.50.88,Gemmer06,Reimann08,Linden09,Bartsch09} it is shown how seemingly irreversible, thermodynamic behavior of macroscopic systems can be explained in the framework of standard Quantum Mechanics and that the approach proposed in \cite{slloydthesis,Borowski03,Popescu05} is capable of explaining the phenomenon of equilibration in a natural way.
\item There are some works that investigate equilibration and thermalization in particular models \cite{Wang08,Devi09,Cramer08,Cramer09,Merkli09}. Due to the additional structure in the less general situations considered in these works a more detailed analysis is possible and the authors can make assertions about the time scales on which equilibration happens.
\item In \cite{Schroeder10} it is shown how the concepts of work and heat can be defined on purely microscopical grounds without using classical external driving and in \cite{Linden09-2} the limits of purely quantum microscopic thermal machines are investigated. See also the references in \cite{Schroeder10,Linden09-2} for works discussing and applying definitions of work and heat based on time dependent Hamiltonians and external driving.
\item In \cite{0907.1267v1,Devi09} it is shown that a slow and continuous evolution of expectation values is typical for large quantum systems.
\item In \cite{vonneumann1929,0907.0108v1} it is shown that large quantum systems typically are in some sense \emph{ergodic} with respect to coarse grained, macroscopic observables. Ref. \cite{vonneumann1929} was criticized in \cite{PhysRev.114.94,Farquhar57}, but recently it was argued that this criticism was unjustified \cite{0907.0108v1}.  
\item There have been attempts to derive the Second~Law of Thermodynamics \cite{Gemmer01,Brandao07} or a statistical H-Theorem \cite{vonneumann1929} for the von~Neumann entropy from Quantum Mechanics and in \cite{Polkovnikov08} (see also the older references 4 and 5 in \cite{PhysRev.108.17}) a different entropy measure, ``microscopic diagonal entropy'', was proposed to overcome the contradiction between microscopic time reversal invariance and the Second~Law.  
\item In addition to the mainly analytical works cited above there exists a quickly increasing amount of numerical works concerned with equilibration and thermodynamic behavior of open Quantum systems confirming the analytical findings \cite{Borowski03,0708.1324v1,0904.1501v1,Gemmer09,Wu09}.
\end{itemize} 

Unfortunately the often mathematically rigorous and far reaching results of these works are almost complete ignored by textbooks on Statistical Mechanics and Thermodynamics, this is true even for the results obtained by von~Neumann in 1930 \cite{vonneumann1929} (an exception is \cite{Gemmer09}).
This situation is unfortunate since some of the results mentioned above address long standing conceptual issues at the very heart of Statistical Mechanics and Thermodynamics.

\begin{figure}[p]
  \centering\footnotesize
  \subfloat[]{%
    \inputTikZ{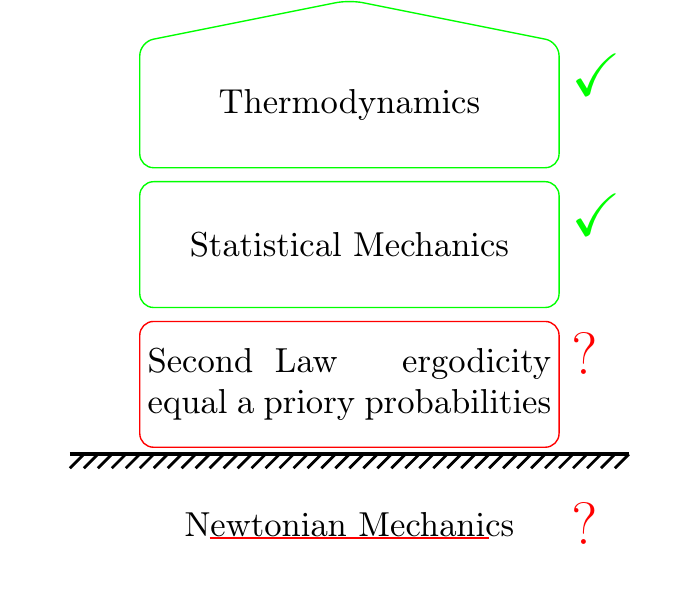}
%     \begin{tikzpicture}[scale=0.71]
%       \path[use as bounding box] (-5,-2) rectangle (5,6.5);%hack to fix wobbling
%       \path[draw=green,fill=white,rounded corners] (0,6.5) -- (-3,5.9) -- (-3,4.1)-- (3,4.1) -- (3,5.9) -- cycle node[anchor=north west] {\color{green}\Large\checkmark};
%       \node[] at (0,5) {Thermodynamics};
%       \path[draw=green,fill=white,rounded corners] (-3,2.1) rectangle (3,3.9) node[anchor=north west] {\color{green}\Large\checkmark};
%       \node[] at (0,3) {Statistical Mechanics};
%       \path[draw=red,fill=white,rounded corners] (-3,0.1) rectangle (3,1.9) node[anchor=north west] {\color{red}\Large ?};
%       \node[text width=4.1cm, text justified] at (0,1) {Second Law \ \  ergodicity equal a priory probabilities};
%       \draw[very thick] (-4,0) -- (4,0);
%       \foreach \x in {-3.8,-3.6,...,4} {\draw[thick] (\x,0) -- (\x-0.2,-0.2);};
%       \node[] at (0,-1) {Newtonian Mechanics};
%       \node[anchor=west] at (3,-1) {\color{red}\Large ?};
%       \draw[red] (-2,-1.2) -- (2,-1.2);
%     \end{tikzpicture}
  }\hfill
  \subfloat[]{%
    \inputTikZ{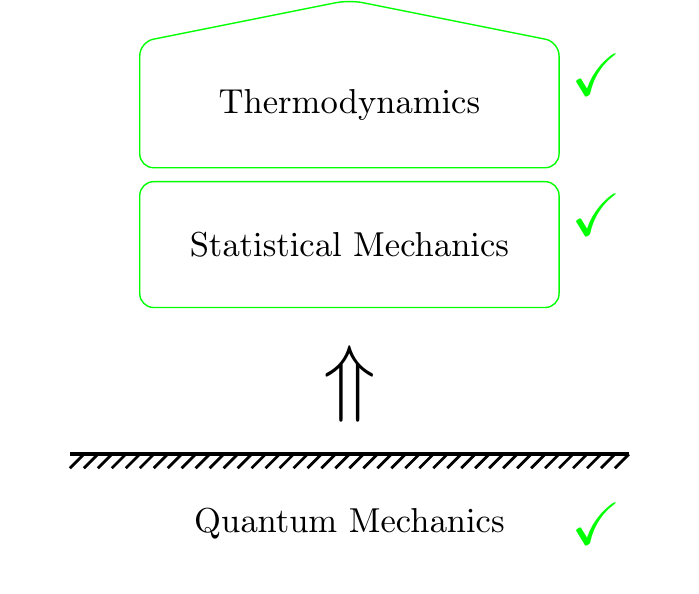}
%     \begin{tikzpicture}[scale=0.71]
%       \path[use as bounding box] (-5,-2) rectangle (5,6.5);%hack to fix wobbling
%       \path[draw=green,fill=white,rounded corners] (0,6.5) -- (-3,5.9) -- (-3,4.1)-- (3,4.1) -- (3,5.9) -- cycle node[anchor=north west] {\color{green}\Large\checkmark};
%       \node[] at (0,5) {Thermodynamics};
%       \path[draw=green,fill=white,rounded corners] (-3,2.1) rectangle (3,3.9) node[anchor=north west] {\color{green}\Large\checkmark};
%       \node[] at (0,3) {Statistical Mechanics};
%       \node[] at (0,1) {\Huge $\Uparrow$};
%       \draw[very thick] (-4,0) -- (4,0);
%       \foreach \x in {-3.8,-3.6,...,4} {\draw[thick] (\x,0) -- (\x-0.2,-0.2);};
%       \node[anchor=west] at (3,-1) {\color{green}\Large\checkmark};
%       \node[] at (0,-1) {Quantum Mechanics};
%     \end{tikzpicture}
  }
  \caption{
    Many of the phenomena correctly described by thermodynamics, which is a mainly phenomenological and very applied theory, can be understood within the framework of Statistical Physics.
    Both theories have a very high degree of corroboration and have proved to be extremely useful.
    In the conventional approach (a) the methods of Statistical Physics are ``derived'' from Newtonian Mechanics and an additional layer of postulates and assumptions that introduce statistical concepts and ensure equilibration.
    These additional assumptions are quite questionable and have provoked quite a lot of debate.
    The irreversibility introduced by postulating the Second Law of Thermodynamics contradicts the time reversal invariance of Newtonian Mechanics and it is still not known whether thermodynamic systems typically are (quasi) ergodic.
    The radical, though natural approach pursued in this work (b) is to replace Newtonian Mechanics by Quantum Mechanics in the hope of getting rid of all extra assumptions.}
  \label{fig:overviewovertheapproach}
\end{figure}

\chapter{Quantum Statistical Mechanics}
\label{sec:quantumstatisticalmechanics}
Especially \cite{slloydthesis,PhysRevA.43.20,PhysRevE.50.88,Popescu05,Popescu06,Gemmer09} argue for a new interpretation of the foundations of Statistical Mechanics.
Following Seth Lloyd \cite{slloydthesis} we called this approach \emph{pure state quantum Statistical Mechanics}.
In what follows we give a concise and self contained review of the results of these and other related works in a unified and consistent notation.
In the first section we introduce the general setup and fix the notation. 
We then review the recent progress in the field and present additional new results concerning the justification of the applicability of the microcanonical and canonical ensemble, equilibration, ergodicity and initial state independence. Finally we show that these results imply a statistical quantum Second~Law of Thermodynamics.

\section{Setup}
\label{sec:setupanddefinitions}
We consider arbitrary quantum systems that can be described using a Hilbert space $\hiH$ of finite dimension $d$.\footnote{If the Hilbert space of a real system is infinite dimensional it should always be possible to find an effective description in a finite dimensional Hilbert space by introducing a high energy cut-off. If eigenstates with extremely high energy had a crucial influence on the behavior of realistic systems physicists would be in a desperate position. Without the ability to prepare and thus study these states in detail it were very difficult to make reliable predictions. The author therefore believes that whenever the behavior of some model is crucially changed by introducing such a cutoff this is due to the very fact that it is a \emph{model}. Moreover, it was demonstrated in \cite{Devi09} that many of the phenomena we that can be rigorously proven in the finite dimensional case also occur in infinite dimensional systems. We thus believe that the restriction to finite dimensions as mainly a technicality.}
We assume that all observables, including energy, are bounded linear operators, i.e have a finite operator norm.

We will often talk about systems that can be divided into two parts, which we will call the bath $B$ and the subsystem $S$, such that $\hiH = \hiH_S \otimes \hiH_B$ where $\hiH_S$ and $\hiH_B$ are the Hilbert spaces of the subsystem and the bath respectively.
It shall be emphasized that we will not make any special a priori assumptions about the size and structure of the bath and system.
All results will be completely general.
The only reason why we call one part the bath and the other the subsystem is that in the end we will be interested in situations where the dimension $d_B$ of the Hilbert space $\hiH_B$ of the bath is much larger than the dimension $d_S$ of the Hilbert space $\hiH_S$ of the system.

We denote by $\mathcal{P}(\hiH)$ the set of all projectors on $\hiH$ and by $\mathcal{P}_n(\hiH)$ the set of all rank $n$  projectors on $\hiH$. 
We write $\ket{\psi}$ and $\ket{\varphi}$ for normalized pure state vectors and use $\psi$ and $\varphi$ to denote their associated pure density matrices in $\mathcal{P}_1(\hiH)$.
The set of all, possibly mixed, normalized density matrices on $\hiH$, i.e. the set of all positive-semidefinite hermitian matrices with trace one, will be denoted by $\mathcal{M}(\hiH)$ and we will use the symbols $\rho$ and $\sigma$ for, possibly mixed, states from $\mathcal{M}(\hiH)$.
Their reduced states, or marginals, on the subsystem and bath are indicated by superscript letters like in $\rho^S = \Tr_B \rho$ and $\rho^B = \Tr_S \rho$.

The Hamiltonian of the joint system $\haH = \haH^\dagger$ has $d$ energy eigenstates $\ket{E_k}$ with corresponding energy eigenvalues $E_k$ that we will assume to be given in units of $\hbar$.
The Hamiltonian governs the time evolution of the joint system.
If the initial state of the system was $\rho_0$ we will denote the state at time $t$ by $\rho_t = U_t\,\rho\,U_t^\dagger$ with $U_t = \ee^{-\iu\,\haH\,t}$.

The Hamiltonians considered herein are completely general except for one extremely weak constraint, namely that they have \emph{non-degenerate energy gaps} or are \emph{non-resonant}.\footnote{This assumption already appears in the work of von~Neumann~\cite{vonneumann1929} and later in \cite{Linden09,Reimann08,0907.0108v1}}
This assumption imposes a restriction on the equality of the gaps between energy eigenvalues, namely
\begin{equation}
  \begin{split}
    E_k - E_l &= E_m - E_n \\\Longrightarrow (k=l \wedge  m=n)\ &\vee\ (k=m \wedge l=n) .
  \end{split}
\end{equation}
Note that there are two slightly different versions of this assumption:
In the first, stronger version the indices $k,l,m,n$ run over all eigenstates of the Hamiltonian, i.e. $k,l,m,n \in \{1,...,d\}$.
This version implies that the spectrum of the Hamiltonian is non-degenerate.
In the weaker version the indices run only over all distinct eigenvalues, so that degeneracies in the energy spectrum are allowed as long as the gaps between the degenerate subspaces are non-degenerate.

It shall be emphasized that even the stronger version is an extremely weak restriction as every Hamiltonian can be made to be non-resonant by adding an arbitrary small random perturbation.
Generic Hamiltonians have non-degenerate energy gaps. 
Every Hamiltonian becomes non-degenerate by adding an arbitrary small random perturbation; therefore the Hamiltonians of macroscopic systems can be expected to satisfy this constraint.

The physical implication of this assumption is that the Hamiltonian is \emph{fully interactive} in the sense that there exists no partition of the composite system into a subsystem and bath such that the Hamiltonian can be written as a sum $\haH = \haH_S \otimes \mathds{1} + \mathds{1} \otimes \haH_B$ where $\haH_S$ and $\haH_B$ act on the subsystem and bath alone.

In the following  we will use the stronger version of the non-degenerate energy gaps assumption for the sake of simplicity.
However, results similar to the ones presented herein hold under the second, weaker version.
Basically, what one has to do is replace projectors onto energy eigenstates $\ketbra{E_k}{E_k}$ by projectors onto degenerate subspaces and refine some of the quantities appearing in the theorems, in particular the \emph{effective dimension} (see the discussion in \cite{0907.1267v1}).

The consequence of the non-degenerate energy gaps assumption that we exploit in the present work is that time averaging a state $\rho_t$ that evolves under such a Hamiltonian 
\begin{equation}
  \expect{\rho_t}_t := \lim_{\tau\to\infty} \frac{1}{\tau} \int_0^\tau \rho_t\,dt .
\end{equation}
gives the same result as dephasing the initial state with respect to the energy eigenbasis of $\haH$
\begin{equation}
  \label{eq:dephasingmap}
  \$[\rho_0] := \sum_k \ketbra{E_k}{E_k} \rho_0 \ketbra{E_k}{E_k} .
\end{equation}
We will therefore use the letter $\omega = \expect{\rho_t}_t = \$[\rho_0]$ to refer to time averaged and dephased states respectively.

In what follows we will often talk about \emph{random pure states} drawn from some subspace $\hiH_R$.
Unless explicitly stated otherwise by a random pure state we mean a state that was chosen according to the Haar measure on $\hiH_R$, which is the unique unitary left and right invariant measure on $\mathcal{P}_1(\hiH_R)$ \cite{duistermaat99} (see appendix~\ref{appendix:thehaarmeasure} for more information).

\section{Ensemble averages and pure state quantum Statistical Mechanics}
\label{sec:purestatequantumstatisticalmechanics}
In conventional Statistical Mechanics probabilities, expectation values, variances and higher moments of observables are computed via ensemble averages.
Depending on the situation under consideration one must employ the microcanonical, canonical or the appropriate grand canonical ensemble \cite{Schwabl02}.
The validity of this approach is beyond all doubt and the results obtained using it have been confirmed by innumerous experiments.

On the other hand, the role of probability \cite{logikderforschung,PhysRev.108.17} in Physics, the problem of ergodicity and especially the microscopic justification of the Second~Law of Thermodynamic are very subtle issues and many fundamental questions concerning them are still open despite many decades of research \cite{UffinkFinal}.

The starting point of our discussion will be to show how the applicability of ensemble averages can be justified using Quantum Mechanics and measure concentration techniques without any extra assumptions.

\subsection{The microcanonical ensemble}
\label{sec:themicrocanonicalensemble}
The microcanonical ensemble is in some sense the most fundamental ensemble.
In classical Statistical Physics it is applied to closed systems in equilibrium.
The other ensembles, canonical and grand canonical can be derived from it \cite{Schwabl02}.

In the quantum setting the microcanonical ensemble is used in situations where all one knows about a closed physical system is that the value of some observable $A$, which corresponds to a conserved quantity, i.e $[\haH,A] = 0$, lies in some interval $I$.\footnote{Note that thermodynamically closed does not necessarily mean completely isolated \cite{Gemmer09}. In this section we will however talk only about completely isolated systems.}
Let $\ket{a}$ be the eigenvectors of $A$ and $\hiH_R$ the restricted subspace spanned by those eigenvectors that have eigenvalues in the interval.
The microcanonical expectation value of any observable $B$ with respect to $\hiH_R$ is then defined to be
\begin{equation}
  \label{eq:microcanonicalexpectationvale}
  \expect{B}_{\mathrm{mc}} = \frac{1}{d_R} \sum_{\ket{a} \in \hiH_R} \bra{a} B \ket{a} = \Tr[\frac{\Pi_R}{d_R} B]
\end{equation}
where $\Pi_R = \sum_{\ket{a} \in \hiH_R} \ketbra{a}{a}$ is the projector onto the subspace $\hiH_R$ of eigenstates of $A$ with eigenvalues in $I$.
Knowing only that measuring $A$ would give a value in $I$ we ascribe to the system the mixed state \footnote{Note that there are other possible generalizations of the microcanonical ensemble to the quantum setting that are discussed in the literature (s. \cite{Bender05,Brody05,Mueller09}).}
\begin{equation}
  \label{eq:microcanonicalstate}
  \rho_{\mathrm{mc}} = \frac{\Pi_R}{d_R} = \frac{1}{d_R} \sum_{\ket{a} \in \hiH_R} \ketbra{a}{a} .
\end{equation}
Equation \eqref{eq:microcanonicalexpectationvale} and \eqref{eq:microcanonicalstate} are the quantum version of the \emph{equal a priory probability postulate}, which is \emph{the} fundamental postulate of convectional Statistical Mechanics.
All compatible states are assigned the same a priory probability.

It is beyond all doubt that this approach to calculate expectation values has proven to be extremely useful and yields results in good agreement with experiments.
However it remains puzzling why dynamically evolving and intrinsically quantum mechanical systems may be described by the static, highly mixed state \eqref{eq:microcanonicalstate}.

\subsubsection{Typicality of general observables}
\label{sec:typicalityofgeneralobservables}
The recent results suggest that the equal a priory probability postulate is dispensable \cite{Popescu06}.
Instead of \emph{assuming} that the state \eqref{eq:microcanonicalstate} yields a good description of the system it is possible to \emph{proof} that for almost all pure states of large systems all subsystems behave \emph{as if the system were} in the state \eqref{eq:microcanonicalstate}.
A statement the authors of \cite{Popescu06} called \emph{General Canonical Principle}.

%\paragraph{Expectation values}
%
The idea to reproduce the results obtained using the microcanonical ensemble average, without added randomness form nothing but pure Quantum Mechanics, and thereby justifying its use, was already discussed in 1991 by J.M.~Deutsch \cite{PhysRevA.43.20}.
A mathematically more precise statement about the equivalence of ensemble averages and expectation values of random pure states can be found in the Ph.D. thesis of Seth Lloyd which appeared in the same year \cite{slloydthesis}:
\begin{theorem}
  \label{theorem:slloydensebleaveragesaresuperflous}
  {\bf \cite{slloydthesis}}
  Let $\hiH_R \subseteq \hiH$ be a subspace of dimension $d_R$ of the Hilbert space $\hiH$ of some physical system.
  Let $\Pi_R$ be the projector onto $\hiH_R$ and let $\expect{\cdot}_\psi$ be the average over random pure states $\psi \in \mathcal{P}_1(\hiH_R)$.
  Then for every observable $B$ with $[B,\Pi_R] = 0$:\footnote{The additional constraint $[B,\Pi_R]=0$ is not discussed in the main text of \cite{slloydthesis}, but it is stated and used in the proof the theorem.}
  \begin{equation}
    \expect{(\Tr[B\,\psi] - \expect{B}_{\mathrm{mc}}])^2}_\psi = \frac{1}{d_R+1} (\expect{B^2}_{\mathrm{mc}} - \expect{B^2}_{\mathrm{mc}}^2) \leq \frac{\|B\|_{\infty}^2}{d_R+1}
  \end{equation}
\end{theorem}

The interpretation of theorem~\ref{theorem:slloydensebleaveragesaresuperflous} is straight forward:
If the dimension $d_R$ of $\hiH_R$ is large, it tells us that the mean square deviation the expectation value of $\Tr[B\,\psi]$ computed over random pure states $\psi \in \mathcal{P}_1(\hiH_R)$ from the microcanonical expectation value $\expect{B}_{\mathrm{mc}}$ is small, which implies that the two expectation values will be similar with high probability. 

The methods used in \cite{Popescu05} to proof the \emph{General Canonical Principle}, namely Levy's lemma (see appendix~\ref{appendix:levyslemma}), can be used to proof a stronger, exponential bound on the probability to observe a deviation from the predictions of the microcanonical ensemble when measuring an observable acting on the full Hilbert space:
\begin{theorem}
  \label{theorem:ensebleaveragesaresuperflousforthemicrocanonicalensemble}
  Let $\hiH_R \subseteq \hiH$ be a subspace of dimension $d_R$ of the Hilbert space $\hiH$ of some physical system.
  The probability that the expectation value of an arbitrary observable $B$ in a randomly chosen pure state $\psi \in \mathcal{P}_1(\hiH_R)$ differs from its microcanonical expectation value with respect to $\hiH_R$ is exponentially small in the sense that for every $\epsilon > 0$
  \begin{equation}
    \probability\left\{ | \Tr[B\,\psi] - \expect{B}_{\mathrm{mc}}| \geq \epsilon \right\} \leq 2\,\ee^{-\frac{C\,d_R\,\epsilon^2}{\|B\|_\infty^2}} ,
  \end{equation}
  where $C$ is a constant with $C = (36\,\pi^3)^{-1}$.
\end{theorem}
\begin{proof}%[Proof of theorem~\ref{theorem:ensebleaveragesaresuperflousforthemicrocanonicalensemble}:]
  The proof is almost completely analogous to a proof in appendix VI of \cite{Popescu05} and relies on Levy's lemma~(s. appendix~\ref{appendix:levyslemma}).
  For an arbitrary fixed observable $B$ we define the function
  \begin{equation}
    f_B(\psi) = \Tr[B\,\psi].
  \end{equation}
  The expectation value $\expect{f_B(\psi)}_\psi$ of this function with respect to a randomly chosen pure states $\psi \in \hiH_R$ clearly is
  \begin{equation}
    \expect{f_B(\psi)}_\psi = \expect{\Tr[B\,\psi]}_\psi = \Tr[B\,\expect{\psi}_\psi] = \Tr[B\,\frac{\Pi_R}{d_R}] = \expect{B}_{\mathrm{mc}} .
  \end{equation}
  Its Lipschitz constant $\eta$ with respect to the Hilbert space norm is upper bounded by $2 \|B\|_\infty$, as \cite{Popescu05}:
  \begin{align}
    | &f_B(\psi_1) - f_B(\psi_2) | = |\Tr[B (\psi_1-\psi_2)]| \nonumber\\
    &\leq \|B\|_\infty\,\|\ket{\psi_1} + \ket{\psi_2}\|_2\,\|\ket{\psi_1} - \ket{\psi_2}\|_2 \\
    &\leq 2\,\|B\|_\infty\,\|\ket{\psi_1} - \ket{\psi_2}\|_2 \nonumber
  \end{align}
  Applying Levy's lemma (see appendix~\ref{appendix:levyslemma}) to $f_B(\psi)$ gives the desired result.
\end{proof}

Theorem~\ref{theorem:ensebleaveragesaresuperflousforthemicrocanonicalensemble} tells us that as $d_R$ becomes large the set of states $\psi$ for which $\Tr[B\,\psi]$ deviates from $\expect{B}_{\mathrm{mc}}$ by at most a given amount becomes exponentially small.
\emph{Typical} states will give expectation values that agree very well with the predictions of the microcanonical ensemble. 

%\paragraph{Variances}
%
Of course, typicality of expectation values is not sufficient to justify the microcanonical ensemble from measure theoretic considerations.
Variances and higher moments also need to be considered.
  
In \cite{slloydthesis} it is claimed that theorem~\ref{theorem:slloydensebleaveragesaresuperflous} implies that not only the expectation values, but in addition all higher moments are likely to be close to the microcanonical ones for typical states.
But what is actually proved is that the variance in state $\psi$ computed with respect to the microcanonical expectation value 
\begin{equation}
  \label{eq:notreallthevarianceinstatepy}
  \Tr[(B - \expect{B}_{\mathrm{mc}})^2\,\psi].
\end{equation}
is close to the microcanonical variance
\begin{equation}
  \sigma_{\mathrm{mc}}^2 = \expect{(B - \expect{B}_{\mathrm{mc}})^2}_{\mathrm{mc}} 
\end{equation}
with high probability given that $d_R$ is large.
The additional deviation caused by the fact that \eqref{eq:notreallthevarianceinstatepy} differs from the variance in state $\psi$
\begin{equation}
  \sigma_{\psi}^2 = \Tr[(B - \Tr[B\,\psi] )^2\,\psi]
\end{equation}
is not taken into account.

But, as one may already anticipate, the additional error typically is very small, so that it is not surprising that theorem~\ref{theorem:ensebleaveragesaresuperflousforthemicrocanonicalensemble} can be used to proof that not only the expectation values, but in addition the variances of almost all states are compatible with the variance of the microcanonical ensemble.
We expect that similar statements hold for all higher moments.

In particular we can proof that:
\begin{theorem}
  \label{theorem:thevarianceinrandompurestatesisthesameasthevariaveofthemcensemble}
  Let $\hiH_R \subseteq \hiH$ be a subspace of dimension $d_R$ of the Hilbert space $\hiH$ of some physical system.
  The probability that the variances of some observable $B$ in a random pure state $\psi \in \mathcal{P}_1(\hiH_R)$
  \begin{equation}
    \sigma_{\psi}^2 = \Tr[(B - \Tr[B\,\psi] )^2\,\psi]
  \end{equation}
  differs from the variance that follows from the microcanonical ensemble
  \begin{equation}
    \sigma_{\mathrm{mc}}^2 = \expect{(B - \expect{B}_{\mathrm{mc}})^2}_{\mathrm{mc}}
  \end{equation}
  is exponentially small, in the sense that for every $\epsilon\geq 0$
  \begin{align}
    &\probability\left\{| \sigma^2_{\psi} - \sigma^2_{\mathrm{mc}} | > \|B\|_\infty^2\,\epsilon \right\} \nonumber\\
    &\leq \min_{0\leq\delta\leq\epsilon} 2\,\ee^{-(C\,d_R\,(\epsilon - \delta))} + 2\,\ee^{-(C\,d_R\,\delta^2)} \\
    &\leq 4 \ee^{-C\,d_R\,(1+2\,\epsilon-\sqrt{1+4\,\epsilon})}
  \end{align}
  where $C$ is a constant with $C = (36\,\pi^3)^{-1}$.
\end{theorem}
\begin{proof}
  Let $\mu^{(n)}_\psi$ and $\mu^{(n)}_{\mathrm{mc}}$ be the $n$-th moment of the probability distribution of the observable $B$ with respect to the state $\psi$ and the microcanonical ensemble respectively, so that in particular $\mu^{(2)}_\psi = \sigma_{\psi}^2$ and $\mu^{(2)}_{\mathrm{mc}} = \sigma_{\mathrm{mc}}^2$. To simplify the notation we define 
  \begin{align}
    \Delta^{(n)}_1 &= | \mu^{(n)}_\psi - \Tr[(B - \expect{B}_{\mathrm{mc}})^n\,\psi] |\\
    \Delta^{(n)}_2 &= | \Tr[(B - \expect{B}_{\mathrm{mc}})^n\,\psi] - \mu^{(n)}_{\mathrm{mc}} | .
  \end{align}
  For all $\epsilon \geq 0$ we have:
  \begin{align}
    &\probability\left\{| \mu^{(n)}_{\psi} - \mu^{(n)}_{\mathrm{mc}} | > \epsilon \right\} \nonumber\\
    &\leq \min_{0\leq\delta\leq\epsilon} \probability\left\{ \Delta^{(n)}_1 \geq \epsilon - \delta \lor \Delta^{(n)}_2 \geq \delta \right\} \\
    &\leq \min_{0\leq\delta\leq\epsilon} \probability\left\{ \Delta^{(n)}_1 \geq \epsilon - \delta \right\}  + \probability\left\{ \Delta^{(n)}_2 \geq \delta \right\}
  \end{align}
  The second term in the last line can be bounded.
  Applying theorem~\ref{theorem:ensebleaveragesaresuperflousforthemicrocanonicalensemble} to $B_n = (B - \expect{B}_{\mathrm{mc}})^n$ gives
  \begin{equation}
    \probability\left\{ \Delta^{(n)}_2 \geq \delta \right\} \leq 2\,\ee^{-\frac{C\,d_R\,\delta^2}{\|B_n\|_\infty^2}} .
  \end{equation}
  This is an exponential version of the bound found in \cite{slloydthesis}.
 
  Bounding the first term is in general more complicated except for the variances where we can use the following argument:
  Assume that the deviation between $\Tr[B\,\psi]$ and $\expect{B}_{\mathrm{mc}}$ is
  \begin{equation}
    \epsilon' = \Tr[B\,\psi] - \expect{B}_{\mathrm{mc}} , 
  \end{equation}
  then 
  \begin{align}
    \sigma_{\psi}^2 &= \Tr[(B - \Tr[B\,\psi] )^2\,\psi] = \Tr[(B - ( \expect{B}_{\mathrm{mc}} + \epsilon' ) )^2\,\psi] \nonumber\\
    &= \Tr[(B - \expect{B}_{\mathrm{mc}})^2\,\psi] + \epsilon'^2 - 2\,\epsilon'\,\Tr[(B - \expect{B}_{\mathrm{mc}})\,\psi] \\
    &= \Tr[(B - \expect{B}_{\mathrm{mc}})^2\,\psi] - \epsilon'^2 ,
  \end{align}
  so that
  \begin{equation}
    |\Tr[B\,\psi] - \expect{B}_{\mathrm{mc}}| \leq \epsilon' \Longrightarrow \Delta^{(2)}_1 \leq \epsilon'^2
  \end{equation}
  and therefore we have by theorem~\ref{theorem:ensebleaveragesaresuperflousforthemicrocanonicalensemble} for all $0 \leq \delta \leq \epsilon$ 
  \begin{align}
    &\probability\left\{ \Delta^{(2)}_1 \geq \epsilon - \delta \right\} \nonumber\\
    &\leq \probability\left\{ | \Tr[\psi\,B] - \expect{B}_{\mathrm{mc}}| \geq \sqrt{\epsilon - \delta} \right\} \\
    &\leq 2\,\ee^{-\frac{C\,d_R\,(\epsilon - \delta)}{\|B\|_\infty^2}} .
  \end{align}
  Combining the two estimates we arrive at:
  \begin{equation}
    \probability\left\{| \sigma^2_{\psi} - \sigma^2_{\mathrm{mc}} | > \epsilon \right\} \leq \min_{0\leq\delta\leq\epsilon} 2\,\ee^{-\frac{C\,d_R\,(\epsilon - \delta)}{\|B\|_\infty^2}} + 2\,\ee^{-\frac{C\,d_R\,\delta^2}{\|(B-\expect{B}_{\mathrm{mc}})^2\|_\infty^2}}
  \end{equation}
  Now, every observable can be renormalized such that $\expect{B}_{\mathrm{mc}} = 0$ and rescaled such that its operator norm is one. Doing this one changes the variance by a factor of $\|B\|_\infty^2$ so that we get
  \begin{equation}
    \probability\left\{| \sigma^2_{\psi} - \sigma^2_{\mathrm{mc}} | > \|B\|_\infty^2\,\epsilon \right\} \leq \min_{0\leq\delta\leq\epsilon} 2\,\ee^{-(C\,d_R\,(\epsilon - \delta))} + 2\,\ee^{-(C\,d_R\,\delta^2)} .
  \end{equation}
  Substituting $\delta = 1/2 (\sqrt{1+4\,\epsilon}-1)$ gives the second bound.
\end{proof}
Note that all important steps in the above discussion are valid also for higher moments except for the bound on $\probability\{ \Delta^{(n)}_1 \geq \delta \}$, which is especially simple for the special case $n=2$.
We expect however that slightly more complicated arguments can be made for all higher moments.

Measuring the same typical pure state of a large enough quantum system we therefore can expect to not only get expectation values that are close to the microcanonical ones but in addition the observed variances will be almost identical to the ones predicted by conventional Statistical Mechanics.
Note that these variances are caused by \emph{objective quantum uncertainties}\footnote{The interpretation of the word \emph{objective} depends on the preferred interpretation of Quantum Mechanics. A discussion of this point (that comes to the conclusion that Quantum Mechanical probabilities are not \emph{objective} in a certain sense) can for example be found in \cite{PhysRev.108.17}. However, they are certainly in some sense \emph{more objective} than probabilities that result form the voluntary dismissal of information due to coarse graining. We shall not elaborate on this point here as it would lead us to far away from the subject of this work.} and not by ensemble averages due to a \emph{subjective lack of knowledge} of the micro state.

Concluding we may say that, given an ensemble of large quantum mechanical systems we are, by measure only a reasonably small number of observables, with very high probability, unable to decide whether all systems of the ensemble are in the same random pure state choose from some subspace, or representatives of the corresponding microcanonical ensemble.
We call this property of large quantum systems \emph{microcanonical typicality}.
However, there are combinations of initial states $\psi_0$ and observables $B$ that give a measurement statistic that deviates radically from the predictions of the microcanonical ensemble.
This happens for example when $\psi_0$ is an eigenstate of $B$.
These measurements are the ones that best characterize the system under consideration and an experimentalist will always seek for such a characterization.
Thus the physical significance of the above results is questionable.

In the following sections we will elaborate more on this point and present arguments similar to theorem~\ref{theorem:ensebleaveragesaresuperflousforthemicrocanonicalensemble} for coarse grained observables and for situations where only a subsystem of a larger quantum system is experimentally accessible and we will see that in these situations the criticism expressed above does not apply.

\subsubsection{Typicality of coarse grained observables}
\label{sec:typicalitycoarsegrainedobservables}
We have seen that when all observables are experimentally accessible there always exist measurements, in particular measurements in the eigenbasis, which give a measurement statistic for a random pure state that deviates radically from the one predicted by the microcanonical ensemble.

However, on macroscopic systems most observables are not accessible.
This is not only a consequence of experimental limitations but manly due to the vast number of dimensions of the Hilbert spaces of macroscopic systems \cite{vonneumann1929,Reimann08}.
As an example consider the spin degrees of freedom of a macroscopic magnet.
The typical Hilbert space of such a system has a dimension of the order of $2^{10^{23}}$.
Trying to measure an observable that can distinguish that many states, or even worse, doing state tomography on such a system, certainly is a completely futile task.

Obviously we need to find a way to take our limited capabilities into account when seeking a realistic description of macroscopic systems.
The way we will do that here is the simplest and most straight forward one can possibly think of and similar considerations date back to the work of von~Neumann \cite{vonneumann1929}.

Let $M = \{M_i\}$ be the set of experimentally accessible \emph{macro observables} $M_i$, where, without loss of generality we can assume that the $M_i$ are positive-semidefinite $M_i \geq 0$ and have trace one $\Tr M_i = 1$.
We think of the $M_i$ as macroscopic observables, so that, due to the limited resolution of our measurement apparatuses, the $M_i$ will be highly degenerate.
Furthermore we want the $M_i$ to be \emph{classical} in the sense that $[M_i,M_j] = 0$.
Such a set $M$ of commuting observables induces a pseudo norm and an associated pseudo trace distance
\begin{equation}
  \tracedistance_M(\rho,\sigma) = \max_{M_i\in M} \Tr[M_i(\rho-\sigma)] 
\end{equation}
which measures how well two states $\rho$ and $\sigma$ can be distinguished from one another by the restricted set of observables.\footnote{Note that $\tracedistance_M(\cdot,\cdot)$ reduces to the normal trace distance if $M = \mathcal{M}(\hiH)$. See appendix~\ref{appendix:distancemeasuresforquantumstates} for more information on distance measures for quantum states.}
The set of accessible measurements partitions the total Hilbert space $\hiH$ of the system into a complete set of $m$ orthogonal subspaces $\{\hiH_r\}$ with $\bigoplus_{r=1}^m \hiH_r = \hiH$ of \emph{macroscopically distinguishable states}, or \emph{macro states}, such that states from one subspace can not be distinguished by any of the $M_i$ and that two states are distinguishable by at least one of the $M_i$ whenever they are in different subspaces:
\begin{align}
  \forall \hiH_r\quad &\forall \rho,\sigma \in \mathcal{P}_1(\hiH_r)  & \tracedistance_M(\rho,\sigma) &= 0 \\
  \forall \hiH_r\neq\hiH_s\ &\forall \rho \in \hiH_r,\sigma \in \mathcal{P}_1(\hiH_s)  & \tracedistance_M(\rho,\sigma) &> 0
\end{align}
Every macroscopic observable $A$ that we can measure by using all our measurement capabilities is of the form
\begin{equation}
  A = \sum_{r=1}^m \alpha_r \Pi_r
\end{equation}
where the $\Pi_r$ are the projectors onto the corresponding subspaces $\hiH_r$ and the $\alpha_r$ real parameters.

In realistic situations we can expect that $m \ll d$ and the following theorem tells us that we are unlikely to have any chance of distinguishing a random pure state from the microcanonical state under these conditions:
\begin{theorem}
  \label{theorem:ensebleaveragesaresuperflousforcoarsegraindobservables}
  Let $\hiH_R \subseteq \hiH$ be a restricted subspace of dimension $d_R$ of the Hilbert space $\hiH$ of some physical system.
  Assume that the physically feasible, macroscopic measurements allow one to distinguish a total number of $m$ macro states.
  Then the probability that a random pure state $\psi \in \mathcal{P}_1(\hiH_R)$ gives an expectation value for any of the accessible macroscopic observables $A$ that differs from that of the microcanonical one with respect to $\hiH_R$ is exponentially small, namely
  \begin{equation}
    \label{eq:allmacroscopicobservablesaretypical}
    \probability\left\{ \max_A | \Tr[A\,\psi] - \expect{A}_{\mathrm{mc}}| \geq \epsilon \right\} \leq 2\,m\,\ee^{-\frac{C\,d_R\,\epsilon^2}{m^2\,\|A\|_\infty^2}} ,
  \end{equation}
  where $C$ is a constant with $C = (36\,\pi^3)^{-1}$.
\end{theorem}
\begin{proof}
  The proof is inspired by the considerations in appendix VI of \cite{Popescu05}.
  As explained above $M$ defines a set of mutually orthogonal projectors $\Pi_r$ onto subspaces of indistinguishable states and consequently every accessible observable is of the form
  \begin{equation}
    A = \sum_{r=1}^m \alpha_r \Pi_r 
  \end{equation}
  so that $\|A\|_\infty = \max_r |\alpha_r|$.
  Obviously for all such observables it holds that
  \begin{align}
    | \Tr[A\,\psi] - \expect{A}_{\mathrm{mc}}| &\leq \sum_{r=1}^m | \alpha_r (\Tr[\Pi_r\,\psi] - \expect{\Pi_r}_{\mathrm{mc}})| \\
    &\leq m\,\max_{r} |\alpha_r| | \Tr[\Pi_r\,\psi] - \expect{\Pi_r}_{\mathrm{mc}} |
  \end{align}
  Inserting $B = \alpha_r\,\Pi_r$ into theorem~\ref{theorem:ensebleaveragesaresuperflousforthemicrocanonicalensemble} we find that for random pure states $\psi \in \mathcal{P}_1(\hiH_R)$
  \begin{equation}
    \probability\left\{ |\alpha_r| | \Tr[\Pi_r\,\psi] - \expect{\Pi_r}_{\mathrm{mc}}| \geq \epsilon \right\} \leq 2\,\ee^{-\frac{C\,d_R\,\epsilon^2}{\alpha_r^2}} ,
  \end{equation}
  where $C = (36\,\pi^3)^{-1}$.
  Using the union bound we see that this implies that 
  \begin{equation}
    \probability\left\{ \exists r:\ |\alpha_r| | \Tr[\Pi_r\,\psi] - \expect{\Pi_r}_{\mathrm{mc}}| \geq \epsilon \right\} \leq 2\,m\,\ee^{-\frac{C\,d_R\,\epsilon^2}{\|A\|_\infty^2}} ,
  \end{equation}
  so that for all accessible observables $A$
  \begin{equation}
    \probability\left\{ | \Tr[A\,\psi] - \expect{A}_{\mathrm{mc}}| \geq m\,\epsilon \right\} \leq 2\,m\,\ee^{-\frac{C\,d_R\,\epsilon^2}{\|A\|_\infty^2}} .
  \end{equation} 
\end{proof}
The important quantity in the above theorem is the quotient $d_R/m^2$ in the exponent of \eqref{eq:allmacroscopicobservablesaretypical} which quantifies how good our abilities to prepare and measure a state are.
Assuming that the dimensions of each of the subspaces of indistinguishable states are approximately identical one can expect that $d_R \approx d/m$ and $d$ grows exponentially with the number of constituents of the system.
In contrast $m$ is basically given by the spread of the spectra of the physically accessible observables divided by the resolution of the measurement apparatuses.
The spread of the spectra can be expected to grow at most polynomial with the system size and the resolution of the measurement apparatuses will be roughly independent of the system size.
One can therefore expect that for large enough systems one enters the regime where $d_R \gg m^2$ and where the above theorem becomes meaningful.

In contrast to theorem~\ref{theorem:ensebleaveragesaresuperflousforthemicrocanonicalensemble}, which we have criticized for being of limited significance, as there always exist observables capable of distinguishing between a random state and the microcanonical state, theorem~\ref{theorem:ensebleaveragesaresuperflousforcoarsegraindobservables} is a statement about \emph{all accessible} observables.

\subsection{The canonical ensemble}
\label{sec:thecanonicalensemble}
The usual situation in which the canonical ensemble is applied are subsystems of weakly interacting composite systems whose total energy is known to lie in some narrow interval.
A slightly more general situation is that of a composite system subject to the constraint that the value of some observable $A$ corresponding to an extensive and conserved quantity is known to lie within some interval.
This understanding of the canonical ensemble includes what is sometimes called the \emph{grand canonical ensemble}.
For the sake of simplicity we restrict ourselves to the canonical case where $A = \haH$. The generalization to the grand canonical case is almost trivial. 

Using the canonical ensemble to calculate expectation values is equivalent to assuming that the state of the system of interest is given by the so called \emph{canonical state}
\begin{equation}
  \label{eq:canonicalstate}
  \rho_c = \frac{1}{Z} \ee^{-\beta\,\haH_S} = \frac{1}{Z} \ee^{-\beta\,E_k} \ketbra{E^S_k}{E^S_k} ,
\end{equation}
where $\beta$ is the inverse temperature, $\ket{E^S_k}$ the eigenstates of the system Hamiltonian $\haH_S$ and 
\begin{equation}
  Z = \Tr \ee^{-\beta\,\haH_S} 
\end{equation}
the \emph{partition sum}, which ensures normalization.

Taking \eqref{eq:canonicalstate} as the system state is usually justified by regarding it as a subsystem of a larger, closed composite system to which the microcanonical ensemble can be applied \cite{Goldstein06,greiner,noltingstatistischephysik01}.\footnote{Alternatively one plead the Bayesian probability and the principle of maximum entropy principle \cite{PhysRev.108.17,PhysRev.106.62}.}
The following is a sketch of how this justification works.

The argument presented herein follows closely the discussion in \cite{Goldstein06}.
Note that the argument is solely based on combinatorics and the identification of the thermodynamic entropy with the entropy defined via the number of compatible micro states.
There is nothing specifically quantum to it.
Very similar arguments can be found in nearly every textbook on Statistical Mechanics.

The Hamiltonian of the composite system
\begin{equation}
  \haH = \haH_S \otimes \mathds{1} + \mathds{1} \otimes \haH_B + \haH_{SB}
\end{equation}
consists of a system Hamiltonian $\haH_S$, a bath Hamiltonian $\haH_B$ and an interaction term $\haH_{SB}$.
The interaction term is assumed to be small in the sense that the total energy of the system is approximately the sum of the system energy and the bath energy, i.e. that energy is extensive, and that the energy eigenstates are close to product states.

The energy of the composite system is assumed to be known to lie in some interval $[E, E+\Delta E]$ that is assumed to be small on a macroscopic energy scale, but still large enough such that the subspace $\haH_R$ spanned by the energy eigenstate with eigenvalues in the interval is large.

Assuming that the composite system is in the microcanonical state and using that the energy eigenstates of $\haH$ are approximately product states we find for the reduced state of the system
\begin{align}
  \label{eq:microcanonicalreducedstate}
  \rho_{\mathrm{mc}}^S &= \Tr_B \rho_{\mathrm{mc}} \\
  \label{eq:microcanonicalreducedstate2}
  &\approx \frac{1}{Z} \sum_{k=1}^{d_S} d_k(E^B)\,\ketbra{E^S_k}{E^S_k} ,
\end{align}
where the $\ket{E^S_k}$ are the eigenstates of $\haH_S$ with energy $E^S_k$ and the $d_k(E^B)$ are the number of eigenstates of $\haH_B$ with eigenvalues in the interval $[E - E^S_k, E - E^S_k + \Delta E]$.

The last step is to introduce the concept of temperature.
The inverse temperature of the bath is defined via $\beta = \partial S(E^B)/\partial E^B$ where $S(E^B)$ is the entropy of the bath when it is held at energy $E^B$.
Assuming that the energy levels of the bath become exponentially dense with increasing energy, which seems to be a reasonable assumption for most thermodynamic systems, one can expect that $S(E^B) \approx \log(d_k(E^B))$.\footnote{This is probably the most critical step in the argument. The assumption of exponentially dense energy gaps conflicts with the assumption that $\haH_{SB}$ does not significantly influence the eigenstates of the uncoupled Hamiltonian $\haH_S \otimes \mathds{1} + \mathds{1} \otimes \haH_B$, as this can be guarantied only when the coupling is smaller than the energy gaps of the uncoupled Hamiltonian.}
Such that, if the bath is much larger than the system we have:
\begin{equation}
  d_k(E^B) \approx \ee^{S(E - E^S_k)} \approx \ee^{S(E) - \beta\,E^S_k} \propto \ee^{-\beta\,E^S_k}
\end{equation}
So that finally one reaches the conclusion that $\rho_{\mathrm{mc}}^S = \Tr_B \rho_{\mathrm{mc}} \approx \rho_{c}$ under the given conditions.\footnote{Using a similar argument, but under additional assumptions on the interaction Hamiltonian, namely that it only couples adjacent energy eigenstates, the canonical ensemble is also derived in \cite{tasaki98}.}

Now the question is: Is it possible to come to the same conclusion without using the ad hoc assumption of the microcanonical state for the composite system?
In \cite{slloydthesis} consequences of theorem~\ref{theorem:slloydensebleaveragesaresuperflous} on the equivalence of expectation values obtained using the canonical ensemble and expectation values of typical quantum states have been already been discussed.
Using similar arguments it is shown in \cite{PhysRev.114.94,tasaki98,Goldstein06} that the reduced state of a typical random state from the subspace compatible with the imposed energy constraint will, with high probability, be close to $\rho_{\mathrm{c}}$.
Herein we focus on the more rigorous exponential bounds provided by theorem~\ref{theorem:ensebleaveragesaresuperflousforthemicrocanonicalensemble} and the results obtained in \cite{Popescu05}.

Of course theorem~\ref{theorem:ensebleaveragesaresuperflousforthemicrocanonicalensemble} is also applicable to observables that act only locally on the subsystem and our considerations concerning variances and higher moments also remain valid.
Consequently theorem~\ref{theorem:ensebleaveragesaresuperflousforthemicrocanonicalensemble} and \ref{theorem:thevarianceinrandompurestatesisthesameasthevariaveofthemcensemble} already tell us that the measurement statistics of local observables does not differ much whether we assume that the composite system is in the microcanonical state corresponding to $\hiH_R$ or in one particular random pure state from $\hiH_R$.

For reduced states of random pure states an even more powerful statement can be proved. 
This is the main result of \cite{Popescu05}:
\begin{theorem}
  {\bf (Theorem 1 in \cite{Popescu05})\footnote{In many situations theorem~\ref{theorem:ensebleaveragesaresuperflousforthecanonicalensemble} can be further improved. See \cite{Popescu05} for details.}}
  \label{theorem:ensebleaveragesaresuperflousforthecanonicalensemble}
  Let $\hiH_R \subseteq \hiH$ be a subspace of dimension $d_R$ of the Hilbert space $\hiH = \hiH_S \otimes \hiH_B$ of some physical system.
  The probability that the reduced state $\rho^S = \Tr_B \psi$ of a randomly chosen pure state $\psi \in \mathcal{P}_1(\hiH_R)$ is more than $\epsilon > 0$ away from the reduced microcanonical state $\rho_{\mathrm{mc}}^S = \Tr_B \rho_{\mathrm{mc}} $ is given by
  \begin{equation}
    \probability\left\{ \tracedistance(\rho^S, \rho_{\mathrm{mc}}^S) \geq 2\,\epsilon + 2\,\sqrt{\frac{d_S}{d^{\mathrm{eff}}_B}} \right\} \leq 2\,\ee^{-C\,d_R\,\epsilon^2} ,
  \end{equation}
  with $C=(18\,\pi^3)^{-1}$ and 
  \begin{equation}
    d^{\mathrm{eff}}_B = d^{\mathrm{eff}}(\rho_{\mathrm{mc}}^B) \geq \frac{d_R}{d_S} .
  \end{equation}
\end{theorem}
Whenever $d_R \gg d_S$, which is exactly the situation we are interested in, this theorem gives a full replacement for the assumption made in \eqref{eq:microcanonicalreducedstate}.
If one trusts the argument presented above that $\rho_{\mathrm{mc}}^S \approx \rho_{c}$, this theorem, together with the usual assumption of weak interaction, proves that almost every pure state drawn from a sufficiently large subspace is locally equivalent to the canonical state.
That is, there exists \emph{no} measurement at all by which they can be distinguished.
This is a measure theoretic justification for the applicability of the canonical ensemble that does not rely on the microcanonical ensemble or the equal a priory probability postulate.
The authors of \cite{Popescu05} call it \emph{General Canonical Principle}.

\section{Average effective dimension of random pure states}
\label{sec:averageeffectivedimensionofrandompurestates}
In this section we will discuss the \emph{effective dimension}
\begin{equation}
  d^{\mathrm{eff}}(\omega) = \frac{1}{\Tr[\omega^2]} ,
\end{equation}
where $\omega = \$[\psi_0] = \expect{\psi_t}_t$, of random pure initial states $\psi_0$ drawn according to different distributions.
This quantity will be important in the following discussion.
Roughly spoken we will find that a high effective dimension causes thermodynamic behavior, while a small effective dimension will make quantum effects observable.

Before we go on it is useful to develop an intuitively understanding for the effective dimension.
Obviously we have $d^{\text{eff}}(\psi) = 1$ if $\psi$ is pure and the completely mixed state has an effective dimension of $d^{\text{eff}}(\mathds{1}/d) = d$.
Expanding an arbitrary pure initial state $\psi_0$ in the energy eigenbasis as follows
\begin{equation}
  \psi_0 = \sum_{kl} c_k\,c_l^*\,\ee^{-\iu(E_k - E_l)t} \ketbra{E_k}{E_l}
\end{equation}
we find that, under the assumption of non-degenerate energy gaps, its effective dimension is
\begin{equation}
  d^{\mathrm{eff}}(\omega) = \frac{1}{\Tr[\$[\psi_0]^2]} = \frac{1}{\sum_k |c_k|^4} .
\end{equation}
Therefrom we see that the effective dimension can be interpreted as a measure for the number of energy eigenstates that contribute significantly to the given initial state $\psi_0$.
This intuition can already serve as a justification for the assumption that for macroscopic objects $d^{\mathrm{eff}}(\omega)$ will typically be very large.

In the remainder of this section we will establish a number of rigorous measure theoretic statements supporting this intuition.
The considerations will necessarily be quite technical.
In particular, we will consider states drawn according to the Haar measure from subspaces of the total Hilbert space, product states, where both tensor components are drawn from subspaces according to the Haar measure, and states from the mean energy ensemble.
When first reading this work it is maybe better to settle with the intuitive argument given above, skip the rest of this section and continue reading in section~\ref{sec:equilibration}.

\subsection{States drawn from subspaces}
\label{sec:equilibrationoftypicalstates}
One of the centrals result derived in \cite{Linden09} is that almost all pure states drawn according to the unitary invariant Haar measure from a high dimensional subspace have a high effective dimension:
\begin{theorem}
  \label{theorem:highaveragedeffectivedimensionisgeneric}
  {\bf (Theorem 2 in \cite{Linden09})}
  i) The average effective dimension with respect to a Hamiltonian with non-degenerate energy gaps $\expect{d^{\mathrm{eff}}(\omega)}_{\psi_0}$, where the average is computed over uniformly random pure initial states $\psi_0 \in \mathcal{P}_1(\mathcal{H_R})$ drawn from some subspace $\hiH_R \in \hiH$ of dimension $d_R$, is such that
  \begin{equation}
    \expect{d^{\mathrm{eff}}(\omega)}_{\psi_0} \geq \frac{d_R}{2} .
  \end{equation}
  ii) For a random pure initial state $\psi_0 \in \mathcal{P}_1(\mathcal{H_R})$, the probability that $d^{\mathrm{eff}}(\omega)$ is smaller than $d_R/4$ is exponentially small, namely
  \begin{equation}
    \probability\left\{d^{\mathrm{eff}}(\omega) < \frac{d_R}{4}\right\} \leq 2\,\ee^{-C\,\sqrt{d_R}}
  \end{equation}
  with a constant $C = \frac{\ln(2)^2}{72\,\pi^3}$.
\end{theorem}

The above theorem states that whenever one draws a state according to the Haar measure form a high dimensional subspace one will almost certainly get a state with a high effective dimension.
Note that theorem~\ref{theorem:highaveragedeffectivedimensionisgeneric} is a very strong statement.
It is actually much stronger than what we will need in the following, namely that $d^{\text{eff}}(\omega)$ is much larger than some low, fixed power of the dimension of the Hilbert space of the subsystem $d_S$.

\subsection{Product states}
\label{sec:equilibrationofproductstates}
A particularly interesting class of initial states are product states.
Theorem~\ref{theorem:highaveragedeffectivedimensionisgeneric} shows that almost all states chosen from sufficiently large subspaces have a high effective dimension.
The set of product states however is not a subspace.

The applicability of theorem~\ref{theorem:highaveragedeffectivedimensionisgeneric} to product states is therefore limited to the case where either the system or the bath states are fixed and the other is chosen from a subspace $\hiH_{S,R}$ or $\hiH_{B,R}$ of the Hilbert space of the bath or system respectively, such that $\hiH_R = \psi^S_0 \otimes \hiH_{B,R}$ or $\hiH_R = \hiH_{S,R} \otimes \psi^B_0$.

Here we show that a slightly modified version of the first part of theorem~\ref{theorem:highaveragedeffectivedimensionisgeneric} holds for product states where both the system and the bath part are chosen from subspaces $\hiH_{S,R}$ and $\hiH_{B,R}$ respectively:
\begin{theorem}
  \label{theorem:highaveragedeffectivedimensionisgenericforproductstates}
  The average effective dimension with respect to a Hamiltonian with non-degenerate energy gaps $\expect{d^{\mathrm{eff}}(\omega)}_{\psi^S_0 \otimes \psi^B_0}$ where the average is computed over product states $\psi^S_0 \otimes \psi^B_0$ consisting of uniformly random pure initial states $\psi^{S/B}_0 \in \mathcal{P}_1(\hiH_{S/B,R})$ chosen from subspaces $\hiH_{S/B,R} \subseteq \hiH_{S/B}$ of dimension $d_{S/B,R}$ respectively is such that
  \begin{equation}
    \expect{d^{\mathrm{eff}}(\omega)}_{\psi^S_0 \otimes \psi^B_0} \geq \frac{(d_{S,R}+1)\,(d_{B,R}+1)}{4} .
  \end{equation}
\end{theorem}
\begin{proof}
  The proof uses some of the ideas from the proof of theorem~2 in \cite{Linden09}.
  The first step is to see that the average effective dimension is bounded by the inverse of the average purity of the time averaged state as follows.
  \begin{equation}
    \expect{d^{\mathrm{eff}}(\omega)}_{\psi^S_0 \otimes \psi^B_0} = \expect{\frac{1}{\Tr[\omega^2]}}_{\psi^S_0 \otimes \psi^B_0} \geq \frac{1}{\expect{\Tr[\omega^2]}_{\psi^S_0 \otimes \psi^B_0}}
  \end{equation}
  To bound the average purity we first use the simple identity
  \begin{equation}
    \label{eq:usefulidentityfortracesofperatorproducts}
    \Tr[A\,B] = \Tr[(A \otimes B)\,\mathbb{S}] ,
  \end{equation}
  where $\mathbb{S}$ is the swap operator of the two tensor components.
  Equation \eqref{eq:usefulidentityfortracesofperatorproducts} can easily be proved by expanding it in a basis.
  \begin{align}
    \Tr[(A \otimes B)\,\mathbb{S}] &= \sum_{kl} \bra{kl}(A \otimes B)\,\mathbb{S}\ket{kl} \\
    &= \sum_{kl} \bra{kl}(A \otimes B)\ket{lk} \\
    &= \sum_{kl} \bra{k}A\ket{l}\,\bra{l}B\ket{k} \\
    &= \sum_{k} \bra{k}A\,B\ket{k} =  \Tr[A\,B] 
  \end{align}
  Second, we need the following lemma, which follows from the representation theory of the unitary group:
  \begin{lemma}
    {\bf \cite{Linden09}}
    Let $\expect{\cdot}_\psi$ be the average over random pure states $\psi \in \mathcal{P}_1(\hiH_R)$ drawn from some subspace $\hiH_R \subseteq \hiH$ of dimension $d_R$. Then
    \begin{equation}
      \expect{\psi \otimes \psi}_\psi = \frac{\Pi_{RR}\,(\mathds{1}+\mathbb{S})}{d_R\,(d_R+1)} ,
    \end{equation}
     where $\Pi_{RR} = \Pi_R \otimes \Pi_R$ and $\Pi_R$ is the projector onto the subspace $\hiH_R$.
  \end{lemma}
  Third, we need the assumption of non-degenerate energy gaps to identify the time average with the dephasing map introduced in \eqref{eq:dephasingmap}.
  In addition we need another linear swap operator $\tilde{\mathbb{S}}$ that is defined via its action on product states,
  \begin{equation}
    \tilde{\mathbb{S}} \ket{sbs'b'} = \ket{ss'bb'}
  \end{equation}
  where $\ket{s},\ \ket{s'} \in \hiH_S$ and $\ket{b},\ \ket{b'} \in \hiH_B$.
  Note that $\tilde{\mathbb{S}}$ is unitary, $\tilde{\mathbb{S}}^2 = \mathds{1}$ and $\|\tilde{\mathbb{S}}\|_\infty = 1$.

  Writing $\ket{k}$ instead of $\ket{E_k}$ for the eigenstates to simplify the notation, the average purity can be written as follows: 
  \begin{align}
    &\expect{\Tr[\omega^2]}_{\psi^S_0 \otimes \psi^B_0} = \expect{\Tr[\$(\psi^S_0 \otimes \psi^B_0)^2]}_{\psi^S_0 \otimes \psi^B_0} = \\
    &= \expect{\Tr[(\$ \otimes \$)(\psi^S_0 \otimes \psi^B_0 \otimes \psi^S_0 \otimes \psi^B_0)\,\mathbb{S}]}_{\psi^S_0 \otimes \psi^B_0} \\
    &= \expect{\Tr[(\$ \otimes \$)(\tilde{\mathbb{S}}\,(\psi^S_0 \otimes \psi^S_0 \otimes \psi^B_0 \otimes \psi^B_0)\,\tilde{\mathbb{S}})\,\mathbb{S}]}_{\psi^S_0 \otimes \psi^B_0} \\
    &= \expect{\Tr[(\$ \otimes \$)(\tilde{\mathbb{S}}\,(\psi^S_0 \otimes \psi^S_0 \otimes \psi^B_0 \otimes \psi^B_0)\,\tilde{\mathbb{S}})\,\mathbb{S}]}_{\psi^S_0 \otimes \psi^B_0} \\
    &= \Tr[(\$ \otimes \$)(\tilde{\mathbb{S}}\,(\expect{\psi^S_0 \otimes \psi^S_0}_{\psi^S_0} \otimes \expect{\psi^B_0 \otimes \psi^B_0}_{\psi^B_0})\,\tilde{\mathbb{S}})\,\mathbb{S}] \\
    &= D^{-1}\,\Tr[(\$ \otimes \$)(\tilde{\mathbb{S}}\,(\Pi_{SS,R}\,(\mathds{1}_{SS} + \mathbb{S}_{SS}) \otimes \Pi_{SS,R}\,(\mathds{1}_{BB} + \mathbb{S}_{BB}))\,\tilde{\mathbb{S}})\,\mathbb{S}]\\
    &= D^{-1}\,\sum_{kl} \Tr[\ketbra{kl}{kl}\,\tilde{\mathbb{S}}\,(\Pi_{SS,R}\,(\mathds{1}_{SS} + \mathbb{S}_{SS}) \otimes \Pi_{BB,R}\,(\mathds{1}_{BB} + \mathbb{S}_{BB}))\,\tilde{\mathbb{S}}\,\ketbra{kl}{kl}\,\mathbb{S}]\\
    &= D^{-1}\,\sum_{kl} \Tr[\ketbra{kl}{lk}]\,\bra{kl}\tilde{\mathbb{S}}\,(\Pi_{SS,R}\,(\mathds{1}_{SS} + \mathbb{S}_{SS}) \otimes \Pi_{BB,R}\,(\mathds{1}_{BB} + \mathbb{S}_{BB}))\,\tilde{\mathbb{S}}\,\ket{kl}\\
    \label{eq:expectedpurityforproductstates}
    &= D^{-1}\,\sum_k \bra{kk}\tilde{\mathbb{S}}\,(\Pi_{SS,R}\,(\mathds{1}_{SS} + \mathbb{S}_{SS}) \otimes \Pi_{BB,R}\,(\mathds{1}_{BB} + \mathbb{S}_{BB}) )\,\tilde{\mathbb{S}}\,\ket{kk}
  \end{align}
  Thereby $\mathds{1}_{SS/BB}$ and $\mathbb{S}_{SS/BB}$ are the identity and the swap operator on the product spaces $\hiH_{S/B} \otimes \hiH_{S/B}$, the $\Pi_{SS/BB,R}$ are the projectors onto the symmetric product of subspaces $\hiH_{S/B,R} \otimes \hiH_{S/B,R}$ respectively and $D = d_{S,R}\,(d_{S,R}+1)\,d_{B,R}\,(d_{B,R}+1)$.

  If the restricted subspaces for both the system and the bath are taken to be the full Hilbert spaces the fact that $\|\tilde{\mathbb{S}}\|_\infty = 1$ and $\|\mathbb{S}_{SS/BB}\|_\infty = 1$ is sufficient to immediately see that 
  \begin{equation}
    \expect{\Tr[\omega^2]}_{\psi^S_0 \otimes \psi^B_0} \leq \frac{4}{(d_S+1)\,(d_B+1)} .
  \end{equation}
  
  To proof the full lemma all that remains is to show that the sum in \eqref{eq:expectedpurityforproductstates} is bounded by $4/(d_{S,R}\,d_{B,R})$.
  The calculation is quite lengthy but most parts are straight forward, therefore we discuss it only briefly.

  The energy eigenstates can be written as linear combinations of product states
  \begin{equation}
    \ket{k} = \sum_{sb} c_{ksb}\,\ket{sb},
  \end{equation}
  where the $\ket{s}$ and $\ket{b}$ form an orthonormal basis of the system and bath Hilbert space, which is chosen such that the first $d_{S/B,R}$ vectors span the restricted subspaces $\hiH_{S/B,R}$.
  Expanding the middle part of \eqref{eq:expectedpurityforproductstates} gives four terms.
  The two symmetric terms, the one without any $\mathbb{S}_{SS}$ or $\mathbb{S}_{BB}$, and the one with both $\mathbb{S}_{SS}$ and $\mathbb{S}_{BB}$ are both equal to
  \begin{equation}
    \label{eq:symmetrcitermcontribution}
    \sum_{k} \sum_{ss'}^{d_{S,R}} \sum_{bb'}^{d_{B,R}} \|c_{ksb}\|^2\,\|c_{ks'b'}\|^2 .
  \end{equation}
  while the two asymmetric terms are equal to 
  \begin{equation}
    \label{eq:asymmetrcitermcontribution}
    \sum_{k} \sum_{ss'}^{d_{S,R}} \sum_{bb'}^{d_{B,R}} c_{ks'b}\,c_{ksb'}\,c_{ksb}^*\,c_{ks'b'}^* .
  \end{equation}
  Both contributions are real and \eqref{eq:asymmetrcitermcontribution} is always smaller or equal than \eqref{eq:symmetrcitermcontribution}.
  This can be seen by using the fundamental inequality
  \begin{equation}
    a\,b^* + b\,a^* \leq |a|^2 + |b|^2
  \end{equation}
  with $a = c_{ksb}\,c_{ks'b'}$ and $b = c_{ks'b}\,c_{ksb'}$, which gives
  \begin{equation}
    \begin{split}
      c_{ks'b}\,c_{ksb'}\,c_{ksb}^*\,c_{ks'b'}^* + c_{ksb}\,c_{ks'b'}\,c_{ks'b}^*\,c_{ksb'}^* \\\leq \|c_{ks'b}\|^2\,\|c_{ksb'}\|^2 + \|c_{ksb}\|^2\,\|c_{ks'b'}\|^2 .      
    \end{split}
  \end{equation}
  Finally the first symmetric term can be bounded as follows:
  \begin{align}
    &\sum_k \bra{kk} \tilde{\mathbb{S}}\,(\Pi_{SS,R}\otimes\Pi_{BB,R})\,\tilde{\mathbb{S}}\ket{kk} \nonumber\\
    = &\sum_k \bra{kk} (\Pi_{S,R}\otimes\Pi_{B,R}\otimes\Pi_{S,R}\otimes\Pi_{B,R})\ket{kk} \\
    = &\sum_k \bra{k} (\Pi_{S,R}\otimes\Pi_{B,R}\ket{k} \bra{k}\Pi_{S,R}\otimes\Pi_{B,R})\ket{k} \\
    \leq &\sum_k \bra{k} (\Pi_{S,R}\otimes\Pi_{B,R}\ket{k} = d_{S,R}\,d_{B,R}
  \end{align}
  This completes the proof.
\end{proof}

First note that if either $d_S = 1$ and thereby $\hiH_R = \psi^S_0 \otimes \hiH_{B,R}$, or $d_B = 1$ and thereby $\hiH_R = \hiH_{S,R} \otimes \psi^B_0$, we recover the result of theorem~\ref{theorem:highaveragedeffectivedimensionisgeneric}.
The new version of theorem~\ref{theorem:highaveragedeffectivedimensionisgeneric} is sightly better than the original one in situations where both the system and the bath state are drawn from subspaces of comparatively large dimension.

\subsection{States from the mean energy ensemble}
\label{sec:statesfromthemeanenergyenesemble}
In theorem~\ref{theorem:highaveragedeffectivedimensionisgeneric} the bound on the probability to get a state with a low effective dimension drops of exponentially.
This raises the hope that the result does not depend on the details of the measure from which the states are drawn and that similar statements hold true for other non-singular measures.
The Haar measure and the unitary invariant ensemble used in both theorem~\ref{theorem:highaveragedeffectivedimensionisgeneric} and \ref{theorem:highaveragedeffectivedimensionisgenericforproductstates} is sometimes criticized for being unphysical.
It is therefore worth considering other more physically motivated ensembles.
In this section we will calculate the average effective dimension in the \emph{mean energy ensemble}.\footnote{The results presented in this section partially originate from a discussion with Markus M\"{u}ller and Jens Eisert in October 2009.}

Without loss of generality we assumed that the Hamiltonian $\haH$ of the system under consideration is positive, has non-degenerate energy gaps and that its eigenvalues are ordered such that $E_k < E_{k+1} \forall k$.
The mean energy ensemble to energy $E$ is defined as the set of normalized pure states $\psi$ with energy expectation value $E$
\begin{equation}
  \label{eq:meanenergyensemble}
  M_E = \{\psi : \Tr[\psi] = 1, \Tr[\haH\,\psi] = E\} .
\end{equation}
In \cite{Brody05,Bender05,Brody07,0909.3175v1} this ensemble was suggested as a natural alternative to the conventional definition of the microcanonical ensemble we discussed in section~\ref{sec:themicrocanonicalensemble}.

Note that the mean energy ensemble is invariant under shifting and rescaling of all involved energies, i.e. the energy $E$ and all eigenvalues $E_k$ of the respective Hamiltonian, whereas the Harmonic mean
\begin{equation}
  E_H = \frac{d}{\sum_k \frac{1}{E_k}}
\end{equation}
is a nonlinear function of the $E_k$ and is therefore not invariant.
By appropriately shifting all energies $E_H$ can be adjusted to all values between the ground state energy $E_0$ and the mean energy $E_{\diameter} = \Tr[\haH]/d$ while at the same time keeping all energies positive \cite{Mueller09}.
It is therefore always possible to shift the energies such that $E \approx E_H$ when $E_0 < E < E_{\diameter}$.

Using this trick it is shown in \cite{Mueller09} that the manifold of states defined \eqref{eq:meanenergyensemble} shows a strong concentration of measure phenomenon and a method to approximately sample states from the mean energy ensemble is derived:
\begin{theorem}
  \label{theorem:approcimatesamplingscemeforthemeanenergyensemble}
  {\bf (Algorithm 21 in \cite{1003.4982})}
  Consider the mean energy ensemble to energy $E$ of a $d$-dimensional quantum system whose Hamiltonian $\haH$ with eigenvectors $\ket{E_k}$ and eigenvalues $E_k$ is assumed to be positive.
  If $E < E_{\diameter} = \Tr[\haH]/d$ so that all energies can be shifted such that $E \approx E_H = d/\sum_k 1/E_k$ and if $E$ satisfies some mild additional constraints (see the original paper \cite{Mueller09} for more details) a state vector 
  \begin{equation}
    \ket{\psi} = \sum_k c_k \ket{E_k}
  \end{equation}
  from the mean energy ensemble at energy $E$ can be sampled to good approximation by choosing the real and imaginary parts of the expansion coefficients $c_k$ independently from normal distributions with variances
  \begin{equation}
    \label{eq:variancestosamplethemeanenergyensemble}
    \sigma_k = \sqrt{\frac{E}{d\,E_k}} .
  \end{equation}
  If the spectrum of $\haH$ fulfills some additional constraints the described procedure becomes exact in the thermodynamic limit $d \to \infty$ (for details see \cite{Mueller09}).
\end{theorem}
 
The sampling method is similar to the sampling procedure for the Haar measure ensemble (see appendix~\ref{appendix:thehaarmeasure}).
But the variances of the normal distributions from which the real and imaginary parts of the expansion coefficients are drawn are now functions of the energy of the respective eigenstate of the Hamiltonian.

We are particularly interested in calculating the average effective dimension in the mean energy ensemble.
Whether or not the average effective dimension is large depends on the structure of the energy spectrum and the energy $E$ to which the mean energy ensemble is taken.

If $E$ is close to $E_0$, the shift to make $E \approx E_H$ is such that after the shift $E_0$ is close to zero.
The variances \eqref{eq:variancestosamplethemeanenergyensemble} and the expected moduli of the expansion coefficients with respect to the energy eigenbasis are then very nonuniform.
This in general leads to a small average effective dimension.
In the extreme case that $E = E_0$ the only state in $M_E$ is the ground state $\ket{E_0}$ and $d^{\mathrm{eff}}(\omega) = 1$.
If on the other hand $E$ is close to $E_{\diameter} = \Tr[\haH]/d$, a large positive shift is necessary to make $E \approx E_H$.
The variances \eqref{eq:variancestosamplethemeanenergyensemble} and the expected moduli of the expansion coefficients are then relatively uniform.
This results in a large average effective dimension.
In the extreme case that the shift is much larger than the spread of the energy spectrum we recover the Haar measure ensemble for which we already know that the average effective dimension is large (theorem~\ref{theorem:highaveragedeffectivedimensionisgeneric}).
The more uniform the shifted energies are, the higher is the average effective dimension.

This is reflected in the following theorem, which establishes estimates for the average effective dimension in the mean energy ensemble:
\begin{theorem}
  \label{theorem:boundsontheaverageeffectivedimensioninthemeanergyensemble}
  Consider a $d$-dimensional quantum system whose Hamiltonian $\haH$, with eigenvalues $E_k>0$, has non-degenerate energy gaps.
  Let $E_{\diameter} = \Tr[\haH]/d$ and assume that $E$ is such that theorem~\ref{theorem:approcimatesamplingscemeforthemeanenergyensemble} can be applied and that the energies have be shifted such that $E \approx E_H = d/\sum_k \frac{1}{E_k}$.
  Then the average purity of the time averaged state $\expect{\Tr[\omega^2]}_{\psi_0 \in M_E}$, where the average is computed over pure initial states $\psi_0$ drawn from the mean energy ensemble at energy $E$, is to good approximation given by 
  \begin{equation}
    \expect{\Tr[\omega^2]}_{\psi_0 \in M_E} \approx \frac{2\,E^2}{d^2} \sum_k \frac{1}{E_k^2} \lessapprox \frac{2}{d} \frac{E_{\diameter}^2}{E_0^2}
  \end{equation}
  and the average effective dimension $\expect{d^{\mathrm{eff}}(\omega)}_{\psi_0 \in M_E}$ is to good approximation lower bounded by
  \begin{equation}
    \expect{d^{\mathrm{eff}}(\omega)}_{\psi_0 \in M_E} \gtrapprox \frac{d^2}{2\,E^2} \frac{1}{\sum_k \frac{1}{E_k^2}} \gtrapprox \frac{d}{2} \frac{E_0^2}{E_{\diameter}^2}.
  \end{equation}
  If the spectrum of $\haH$ fulfills some additional constraints both statements become exact in the thermodynamic limit $d \to \infty$ (for details see \cite{Mueller09}).
\end{theorem}
\begin{proof}
  The bound on the average effective dimension in the mean energy ensemble follows from the estimate of the average purity of the time averaged state as
  \begin{equation}
    \label{eq:expectationvalueofdeffinverioninequality}
    \expect{d^{\mathrm{eff}}(\omega)}_{\psi_0 \in M_E} = \expect{\frac{1}{\Tr[\omega^2]}}_{\psi_0 \in M_E} \geq \frac{1}{\expect{\Tr[\omega^2]}_{\psi_0 \in M_E}} .
  \end{equation}
  As the Hamiltonian has non-degenerate energy gaps the average purity of the time averaged state is
  \begin{equation}
    \expect{\Tr[\omega^2]}_{\psi_0 \in M_E} = \sum_k \expect{|c_k|^4}_{\psi_0 \in M_E} .
  \end{equation}
  According to theorem~\ref{theorem:approcimatesamplingscemeforthemeanenergyensemble} we can sample from the mean energy ensemble to good approximation by choosing the real and imaginary parts $a_k$ and $b_k$ of the expansion coefficients $c_k = a_k + \iu\,b_k$ from normal distributions with variances $\sigma_k = \sqrt{\frac{E}{d\,E_k}}$. Therefore
  \begin{align}
    \expect{|c_k|^4}_{\psi_0 \in M_E} &\approx \left(\frac{1}{2\,\sqrt{2\,\pi\,\sigma_k^2}}\right)^2 \int_{-\infty}^\infty \int_{-\infty}^\infty (a_k^2+b_k^2)^2 \ee^{-\frac{a_k^2}{2\,\sigma_k^2}} \ee^{-\frac{a_k^2}{2\,\sigma_k^2}} da_k\,db_k \\
    &= 2\,\sigma_k^4 = \frac{2\,E^2}{d^2\,E_k^2} ,
  \end{align}
  and we find that\footnote{\label{footnote:unitaryinvariantensembleisaspecialcaseofthemeanenergyensembel}Note that in the limit $d\to\infty$, where the sampling procedure becomes exact, and if $E$ and all the $E_k$ are identical we recover the first part of theorem~\ref{theorem:highaveragedeffectivedimensionisgeneric} for $d_R = d$. The unitary invariant ensemble for the full Hilbert space is a special case of the mean energy ensemble}  
  \begin{equation}
    \expect{\Tr[\omega^2]}_{\psi_0 \in M_E} \approx \frac{2\,E^2}{d^2} \sum_k \frac{1}{E_k^2} .
  \end{equation}
  To prove the second inequalities we use the fact that the harmonic mean is upper bounded by the arithmetic mean, which follows from the generalized means inequality \cite{sd268119982}:
  \begin{equation}
    E \approx E_H = \frac{d}{\sum_k \frac{1}{E_k}} \leq \frac{1}{d} \sum_k E_k = E_{\diameter}
  \end{equation}
  Using theorem~\ref{theorem:boundsontheaverageeffectivedimensioninthemeanergyensemble} and that $E \approx E_H $ we get
  \begin{equation}
    \expect{\Tr[\omega^2]}_{\psi_0 \in M_E} \approx \frac{2\,E_H^2}{d^2} \sum_k \frac{1}{E_k^2} \leq \frac{2\,E_{\diameter}^2}{d^2} \sum_k \frac{1}{E_k^2} \leq \frac{2}{d} \frac{E_{\diameter}^2}{E_0^2} .
  \end{equation}
\end{proof}

Concluding we can say that as long as $E$ is comparatively high, such that after shifting the energy levels $E_0$ is not too many orders of magnitude lower than $E_{\diameter}$ the average effective dimension will be high.\footnote{Note that \emph{high} will usually mean that $d^{\mathrm{eff}(\omega)}$ is much larger than the dimension $d_S$ of some small subsystem (see section~\ref{sec:subsystemequilibration}, \ref{sec:speedoffluctuationsaroudequilibrium} and \ref{sec:equilibrationandeinselection}). Keeping in mind that $d$ grows exponentially with the number of constituents of the system $2\,E_{\diameter}^2/(d\,E_0^2)$ will be large compared to $d_S$ even if $E_0$ is several orders of magnitude smaller than $E_{\diameter}$.}
The closer $E$ is to the ground state energy, the smaller is the average effective dimension.
This is not surprising.
Lowering the energy we expect to observe a transition from thermodynamic to quantum behavior.
This is precisely what happens, for high $E$ we get a high effective dimension, which, as we will see later, causes thermodynamic behavior, while for lower and lower $E$ the effective dimension will decrease making quantum effects observable.

\section{Equilibration}
\label{sec:equilibration}
One of the most obvious features of thermodynamic systems is the tendency to evolve towards equilibrium.
It is therefore not surprising that the oldest and best understood part of Thermodynamics and Statistical Mechanics is concerned with systems in thermal equilibrium.
The tendency to equilibrate is postulated in the Second Law of Thermodynamics.
Starting from this postulate one can use ensemble theory or the condition of detailed balance to derive equilibrium properties of physical systems.
How, and under which conditions, the microscopic, time reversal invariant dynamics of such systems leads to equilibration and thermalization remains unexplained.

In a time reversal invariant theory equilibration in the usual sense is impossible.\footnote{At least in finite dimensional systems \cite{PhysRev.107.33}.}
We therefore use an extended notion of equilibration and say that a system is in equilibrium when its density matrix stays close to some state, for almost all times and say that it evolves towards equilibrium if it approaches such a state, and then stays close to it, when started in a state far from equilibrium.
Likewise we will say that an observable gives the impression of equilibration when its measurement statistics is compatible with the assumption of an equilibrated system.
This is the case if its expectation value and higher moments are nearly stationary for almost all times.

\subsection{Equilibration of expectation values}
\label{sec:equilibrationofexpectationvalues}
Under which conditions observables can create the impression of equilibration was recently investigated by Peter~Reimann in \cite{Reimann08}.
The main result of this paper is the following, very useful theorem which we present here in the form given in \cite{0907.1267v1}:
\begin{theorem}
  {\bf \cite{Reimann08,0907.1267v1}}
  \label{theorem:expectationvaluesimmilartotimeaverage}
  Let $A$ be an observable and let $\rho_t$ evolve under a Hamiltonian with non-degenerate energy gaps, then
  \begin{equation}
    \expect{(\Tr[A\,\rho_t] - \Tr[A\,\omega])^2}_t \leq \frac{\|A\|_\infty^2}{d^{\mathrm{eff}}(\omega)} ,
  \end{equation}
  where $\omega = \expect{\rho_t}_t$.
\end{theorem}
A similar result is derived in \cite{Gemmer09}.

Theorem~\ref{theorem:expectationvaluesimmilartotimeaverage} is a very remarkable result.
Whenever the effective dimension $d^{\mathrm{eff}}(\omega)$ is large, the time average of the square deviation of the expectation value of any observable from its time average will be small.
Therefore, systems which are in a state with a high effective dimension will look like they were in equilibrium most of the time although in reality they evolve unitarily.
Theorem~\ref{theorem:expectationvaluesimmilartotimeaverage} shows that a time reversal invariant theory can create the impression of equilibration.

It shall be stressed that theorem~\ref{theorem:expectationvaluesimmilartotimeaverage} is a statement about the \emph{dynamics} of states with a high effective dimension.
It is crucial to note, that theorem~\ref{theorem:expectationvaluesimmilartotimeaverage} it is a much stronger statement than the usual typicality arguments often made in Statistical Mechanics. 
Such arguments state that there is a large set of equilibrium states and that one can therefore \emph{expect} that starting in a non-equilibrium state not in this set, one will sooner or later end up in an equilibrium state.
In contrast, theorem~\ref{theorem:expectationvaluesimmilartotimeaverage} implies that initial states which are out of equilibrium, i.e states for which $\Tr[A\,\rho_0]$ is far from $\Tr[A\,\omega]$ \emph{definitely} will equilibrate whenever $d^{\mathrm{eff}}(\omega)$ is large.
It does however not make an assertion about how long it takes to reach equilibrium.
We will come back to this problem in section~\ref{sec:equilibrationtime}.

\subsection{Subsystem equilibration}
\label{sec:subsystemequilibration}
An even stronger result can be obtained for subsystems of large quantum mechanical systems.
Very recently it has been shown in \cite{Linden09} that the dynamics of almost every large quantum system is such that for almost every pure initial state every small subsystem equilibrates.
The main result of a recent work of Noah~Linden et al. \cite{Linden09} is a rigorous bound on the expectation value of the trace distance of the reduced state of the subsystems from its time average in terms of the effective dimension $d^{\mathrm{eff}}(\omega)$:
\begin{theorem}
  \label{theorem:distancefromtimeaverage}
  {\bf (Theorem 1 in \cite{Linden09})}
  Consider any pure state $\psi_t$ evolving under a Hamiltonian with non-degenerate energy gaps.
  Then the average distance between $\rho^S_t = \Tr_B \psi_t$ and its time average $\omega^S = \expect{\rho^S_t}_t$ is bounded by
  \begin{equation}
    \expect{\tracedistance(\rho^S_t,\omega^S)}_t \leq \frac{1}{2} \sqrt{\frac{d_S}{d^\mathrm{eff}(\omega^B)}} \leq \frac{1}{2} \sqrt{\frac{d_S^2}{d^\mathrm{eff}(\omega)}}
  \end{equation}
\end{theorem}

Again it is of utter importance to understand that theorem~\ref{theorem:distancefromtimeaverage} is a statement about the \emph{dynamics} of states with a high effective dimension and therefore much stronger than a typicality argument. 
It implies that initial states which are out of equilibrium, i.e states for which $\rho^S_0$ is far from $\omega^S$ \emph{definitely} will equilibrate towards $\omega^S$ whenever $d^{\mathrm{eff}}(\omega)$ is large.
Again it is difficult to make assertions about the time scales on which equilibration happens (see section~\ref{sec:equilibrationtime}).

\begin{figure}[tb]
  \centering
  \inputTikZ{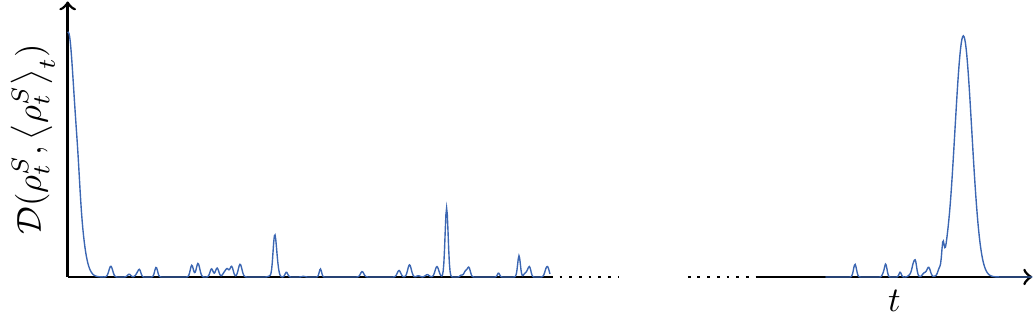}
%   \footnotesize
%   \begin{tikzpicture}[scale=0.7]
%     % \draw[thin] (0,0) rectangle (14,4);
%     \draw[-,thick] (0,0) -- (7,0);
%     \draw[-,dotted,thick] (7,0) -- (8,0) ;
%     \draw[-,dotted,thick] (9,0) -- (10,0) ;
%     \draw[->,thick] (10,0) -- (14,0) node[midway,below] {$t$} ;
%     \draw[->,thick] (0,0) -- (0,4) node[midway,above,rotate=90] {$\tracedistance(\rho^S_t,\expect{\rho^S_t}_t)$} ;
%     \draw[color=structure] plot file {distanceevolution.dat};
%     \draw[color=structure] plot file {distanceevolution2.dat};
%   \end{tikzpicture}
  \caption{Plot of how the time evolution of the trace distance can be imagined.
    Starting in a state far from equilibrium the subsystem will evolve towards states close to the equilibrium state (see section~\ref{sec:equilibrationtime} for more information about the time scales on which equilibration happens.). It will then stay close to the equilibrium state for almost all times. Occasionally fluctuations will drive it out of equilibrium but these events are extremely rare. After an extremely long time the system recurs to its initial state but the time scale on which this happens is enormously large for macroscopic systems.}
  \label{fig:timeevoutionofthetracedistance}
\end{figure}

Of course this theorem only makes sense as long as the Hilbert spaces involved are finite dimensional.
Only then can their dimension serve as a sensible measure for smallness of the subsystem.
However, in \cite{Cramer08,Cramer09,Devi09} it has been demonstrated that small subsystems of quantum systems with infinite dimensional Hilbert spaces also tend to equilibrate.
The example considered in \cite{Cramer08,Cramer09} is a bosonic chain with quadratic coupling and in \cite{Devi09} a system consisting of oscillators coupled with a harmonic interaction Hamiltonian is investigated.
In both works the measure of smallness of the subsystem is the number of units that constitute the subsystem.
It is shown that all small subsystems equilibrate for squeezed pure initial product states while the whole system undergoes a unitary time evolution.
Further numerical studies that confirm the analytical results presented above can be found for example in \cite{Barthel08,Gemmer09}.

\subsection{Equilibration of the purity}
\label{sec:equilibrationofthepurity}
To further illustrate the phenomenon of subsystem equilibration we look at the purity of the subsystem state $p(\rho^S_t) = Tr[(\rho^S_t)^2]$.
A necessary, though not sufficient, condition for equilibration is that the time average of the purity $\expect{p(\rho^S_t)}_t$ and the purity of the time averaged state $p(\omega^S)$ are almost identical
\begin{equation}
  \expect{p(\rho^S_t)}_t \approx p(\omega^S) .
\end{equation}
Their distance can be bounded as follows:
\begin{theorem}
  \label{theorem:purityequilibrates}
  In a system evolving under a Hamiltonian with non-degenerate energy gaps the difference of the time average of the purity $\expect{p(\rho^S_t)}_t$ and the purity of the time averaged state $p(\omega^S)$ is bounded by
  \begin{equation}
    \left| \expect{p(\rho^S_t)}_t - \Tr_S[\omega_S^2] \right| \leq \Tr_B[{\omega^B}^2] + 2\,\Tr[\omega^2] \leq \frac{d_s+2}{d^{\mathrm{eff}}(\omega)} .
  \end{equation}
\end{theorem}
\begin{proof}
  Expanding the initial state in the energy eigenbasis we can write the purity as
  \begin{equation}
    p(\rho^S_t) = \sum_{klmn} c_k\,c_l^*\,c_m\,c_n^*\,\ee^{-\iu(E_k-E_l+E_m-E_n)t} \Tr_S[\Tr_B\ketbra{E_k}{E_l}\,\Tr_B\ketbra{E_m}{E_n}] .
  \end{equation}
  The sum over $k,l,m,n$ can be split up into four parts as follows:
  \begin{equation}
    \sum_{klmn} \dots = \sum_{k \neq l,m \neq n} \dots + \sum_{k \neq l,m = n} \dots + \sum_{k = l,m \neq n} \dots + \sum_{k = l,m = n}
  \end{equation}
  The second and third sum contain only oscillating terms for which the time average $\expect{\dots}_t$ vanishes.
  The time average of the first sum is positive and equal to
  \begin{equation}
    \label{eq:equilibrationofthepurity1}
    \begin{split}
      \sum_{k,m} |c_k|^2\,|c_m|^2\,\Tr_B[\Tr_S\ketbra{E_k}{E_k}\,\Tr_S\ketbra{E_m}{E_m}] \\
      = \Tr_B[{\omega^B}^2] - \sum_k |c_k|^4 \Tr_S[(\Tr_B \ketbra{E_k}{E_k})^2] ,
    \end{split}
  \end{equation}
  and the fourth sum contains only terms that are time independent and is thus equal to
  \begin{equation}
    \label{eq:equilibrationofthepurity2}
    \begin{split}
    \sum_{k,m} |c_k|^2\,|c_m|^2\,\Tr_S[\Tr_B\ketbra{E_k}{E_k}\,\Tr_B\ketbra{E_m}{E_m}] \\
    = \underbrace{\Tr_S[{\omega^S}^2]}_{p(\omega^S)} - \sum_k |c_k|^4 \Tr_B[(\Tr_S \ketbra{E_k}{E_k})^2] .      
    \end{split}
  \end{equation}
  In the derivation of both \eqref{eq:equilibrationofthepurity1} and \eqref{eq:equilibrationofthepurity2} we have used equation (A5) from appendix A in \cite{Linden09}.
  The sums in the right hand side of \eqref{eq:equilibrationofthepurity1} and \eqref{eq:equilibrationofthepurity2} are both bounded by $\Tr[\omega^2]$ so that we find:
  \begin{equation}
    \left| \expect{p(\rho^S_t)}_t - p(\omega^S) \right| \leq  \Tr_B[{\omega^B}^2] + 2\,\Tr[\omega^2]
  \end{equation}
  Using the fact that
  \begin{equation}
    \Tr_B[{\omega^B}^2] \leq \frac{d_s}{d^{\mathrm{eff}}(\omega)}
  \end{equation}
  gives the second bound.
\end{proof}

As we would have already expected from theorem~\ref{theorem:expectationvaluesimmilartotimeaverage} and \ref{theorem:distancefromtimeaverage} we find that the purity is close to a typical value, namely $\Tr_S[{\omega^S}^2]$, most of the time whenever $d^{\mathrm{eff}}$ is large.

\section{Ergodicity}
\label{sec:ergodicity}
The question whether physical systems are (quasi) ergodic plays a central role in all Gibbs like attempts to justify the methods of Statistical Mechanics from Newtonian Mechanics \cite{RevModPhys.27.289}.
Ergodicity is either used directly to identify time and ensemble averages, or as a way to justify the choice of \emph{a priory probabilities} and the microcanonical ensemble.
The question whether all, and if this is not true than which thermodynamic systems are quasi ergodic was investigated by many authors. 
In the classical setup the problem was reduced to the problem of showing \emph{metrical transitivity}, and quasi ergodicity is proven for so called \emph{Kanonische Normalsysteme}.
Irrespective of these efforts the problem still awaits a full solution, so that quasi ergodicity is rather a hypothesis than anything close to a stable foundation for a physical theory.

The approach towards the foundations of Quantum Mechanics we follow herein does not depend on quasi ergodicity.
Nevertheless, due to its historical importance, ergodicity is a property that deserves investigation in its own right.
It turns out that the approach based on measure concentration techniques can be used to \emph{prove} ergodicity:
\begin{theorem}
  \label{theorem:ergodicityistypical}
  Let $\hiH_R \subseteq \hiH$ be a subspace of dimension $d_R$ corresponding to a microcanonical constraint.
  The probability that the time average $\expect{\Tr[B\,\psi_t]}_t$ of the expectation value of an arbitrary observable $B$ computed for a randomly chosen pure initial state $\psi_0 \in \mathcal{P}_1(\hiH_R)$ differs from its microcanonical expectation value with respect to $\hiH_R$ is exponentially small in the sense that for every $\epsilon > 0$
  \begin{equation}
    \probability\left\{ | \expect{\Tr[B\,\psi_t]}_t - \expect{B}_{\mathrm{mc}}| \geq \epsilon \right\} \leq 2\,\ee^{-\frac{C\,d_R\,\epsilon^2}{\|\$[B]\|_\infty^2}} ,
  \end{equation}
  where $C$ is a constant with $C = (36\,\pi^3)^{-1}$ and
  \begin{equation}
    \$[B] = \sum_k \ketbra{E_k}{E_k} B \ketbra{E_k}{E_k} .
  \end{equation}
\end{theorem}
\begin{proof}
  The proof is completely analogous to the proof of theorem~\ref{theorem:ensebleaveragesaresuperflousforthemicrocanonicalensemble} and relies on Levy's lemma.
  For an arbitrary fixed observable $B$ we define the function
  \begin{equation}
    f_B(\psi_0) = \expect{\Tr[B\,\psi_t]}_t = \Tr[B\,\$[\psi_0]].
  \end{equation}
  As $\hiH_R$ corresponds to a microcanonical constraint the projector $\Pi_R$ on $\hiH_R$ commutes with the Hamiltonian $[\Pi_R,\haH]= 0$.
  Thus the expectation value $\expect{f_B(\psi)}_{\psi_0}$ of this function with respect to a randomly chosen pure initial states $\psi_0 \in \mathcal{P}_1(\hiH_R)$ is
  \begin{equation}
    \expect{f_B(\psi_0)}_{\psi_0} = \expect{\Tr[B\,\$[\psi_0]]}_{\psi_0} = \Tr[B\,\$[\expect{\psi_0}_{\psi_0}]] = \expect{B}_{\mathrm{mc}} .
  \end{equation}
  Its Lipschitz constant $\eta$ with respect to the Hilbert space norm is upper bounded by $2 \|\$[B]\|_\infty$, as \cite{Popescu05}:
  \begin{align}
    | &f_B(\psi_1) - f_B(\psi_2) | = |\Tr[B \$[\psi_1-\psi_2]]| = |\Tr[\$[B] (\psi_1-\psi_2)]| \\
    &= \frac{1}{2} |(\bra{\psi_1}+\bra{\psi_2})\,\$[B]\,(\ket{\psi_1}-\ket{\psi_2}) + (\bra{\psi_1}-\bra{\psi_2})\,\$[B]\,(\ket{\psi_1}+\ket{\psi_2})|\\
    &\leq \|\$[B]\|_\infty\,\|\ket{\psi_1} + \ket{\psi_2}\|_2\,\|\ket{\psi_1} - \ket{\psi_2}\|_2 \\
    &\leq 2\,\|\$[B]\|_\infty\,\|\ket{\psi_1} - \ket{\psi_2}\|_2 
  \end{align}
  Applying Levy's lemma (s. appendix~\ref{appendix:levyslemma}) to $f_B(\psi_0)$ gives the desired result.
\end{proof}
Note that for time independent observables $[\haH,B] = 0$, and consequently $\$[B] = B$ so that we recover theorem~\ref{theorem:ensebleaveragesaresuperflousforthemicrocanonicalensemble}.
If an observable is time dependent then $\|\$[B]\|_\infty \leq \|B\|_\infty$ and the bound on the deviation is tighter than the bound of theorem~\ref{theorem:ensebleaveragesaresuperflousforthemicrocanonicalensemble}.

\section{Dynamics of the state of the subsystem}
\label{sec:dynamicsofthestateofthesubsystem}
In the previous sections we have shown under which conditions expectation values and reduced subsystems of large quantum mechanical system equilibrate.
We have however left out a crucial point, namely the timescales on which equilibration happens. 
This section will be concerned with the dynamical properties of the state of the subsystem.

\subsection{Speed of fluctuations around equilibrium}
\label{sec:speedoffluctuationsaroudequilibrium}
Knowing under which conditions the state of the subsystem equilibrates a natural question is: How fast will the fluctuations around the equilibrium state typically be? 
This question was investigated very recently by Noah~Linden et al. \cite{0907.1267v1}.

The first step is to introduce a physically meaningful notion of \emph{speed}.
This is achieved by setting \cite{0907.1267v1}
\begin{equation}
  v(t) = \lim_{\delta t \to 0} \frac{\tracedistance(\rho_t,\rho_{t+\delta t})}{\delta t} = \frac{1}{2} \left\| \frac{d\rho_t}{dt} \right\|_1 ,
\end{equation}
where according to the von~Neumann equation
\begin{equation}
  \frac{d\rho_t}{dt} = \iu\,[\rho_t,\haH] .
\end{equation}
Equivalently one defines the speed of the state of the subsystem as
\begin{equation}
  v_S(t) = \lim_{\delta t \to 0} \frac{\tracedistance(\rho^S_t,\rho^S_{t+\delta t})}{\delta t} = \frac{1}{2} \left\| \frac{d\rho^S_t}{dt} \right\|_1 ,
\end{equation}
with 
\begin{equation}
  \frac{d\rho^S_t}{dt} = \iu\,\Tr_B[\rho_t,\haH] .
\end{equation}
As the choice of the origin of the energy scale does not influence the speed it is convenient to split up the Hamiltonian of the system as follows
\begin{equation}
  \label{eq:generalhamiltonianwithconstanttermabsorbed1}
  \haH = \haH_0 + \haH_S \otimes \mathds{1} + \mathds{1} \otimes \haH_B + \haH_{SB} .
\end{equation}
Thereby $\haH_0$ is taken to be proportional to the identity and $\haH_S$, $\haH_B$ and $\haH_{SB}$ are traceless.\footnote{Note that this decomposition is not unique and the freedom can be used to optimize the quantities appearing in the following theorems.}

Using this it is shown in \cite{0907.1267v1} that:
\begin{theorem}
  \label{theorem:averagespeedisslow}
  {\bf \cite{0907.1267v1}}
  For every pure initial state $\psi_0$ of a composite system evolving under a Hamiltonian of the form of \eqref{eq:generalhamiltonianwithconstanttermabsorbed1} and with non-degenerate energy gaps, it holds that
  \begin{equation}
    \expect{v_S(t)}_t = \frac{1}{2} \left\| \frac{d\rho^S_t}{dt} \right\|_1 \leq \| \haH_S \otimes \mathds{1} + \haH_{SB}\|_\infty \sqrt{\frac{d_S^3}{d^{\mathrm{eff}}(\omega)}} ,
  \end{equation}
  where $\rho^S_t = \Tr_B \psi_t$ and $\omega = \expect{\psi_t}_t$.
\end{theorem}
\begin{proof}
  We will only give a short sketch of the proof herein, for the full proof see the original article \cite{0907.1267v1}.
  The first step is to show that the speed of the subsystem state can be written as
  \begin{equation}
    \label{eq:expectedvelocityintermsofcoefficeints}
    \frac{d\rho^S_t}{dt} = \sum_{k=1}^{d_S^2} c_k(t)\,e_k
  \end{equation}
  such that 
  \begin{equation}
    \label{eq:expectedspeedintermsofcoefficeints}
    \left\langle \left\| \frac{d\rho^S_t}{dt}\right\|_2^2 \right\rangle_t = \sum_k \expect{(c_k(t))^2}_t
  \end{equation}
  where the $d_S^2$ operators $e_k$ form an orthonormal basis for the set of hermitian operators on the Hilbert space of the subsystem such that $\Tr[e_k\,e_l] = \delta_{kl}$ and
  \begin{equation}
    \label{eq:coefficeintsforexpectedvelocity}
    c_k(t) = \Tr\big[\rho_t\,\iu\,[\haH_S + \haH_{SB},e_k\otimes \mathds{1}]\big] .
  \end{equation}
  One then applies theorem~\ref{theorem:expectationvaluesimmilartotimeaverage} from \cite{Reimann08} to $A = \iu\,[\haH_S + \haH_{SB},e_k\otimes \mathds{1}]$ to bound the expectation value of the squared coefficients $\expect{(c_k(t))^2}_t$.
  The final step is to use a standard bound connecting the Hilbert-Schmidt norm used in \eqref{eq:expectedspeedintermsofcoefficeints} with the trace norm. 
\end{proof}

From theorem~\ref{theorem:highaveragedeffectivedimensionisgeneric}, \ref{theorem:highaveragedeffectivedimensionisgenericforproductstates} and the discussion in section~\ref{sec:averageeffectivedimensionofrandompurestates} we know that the effective dimensions $d^{\mathrm{eff}}(\omega)$ typically is very large in realistic situations.
In particular, as all dimensions grow exponentially with the number of constituents of the system it will usually be much larger than any fixed power of $d_S$.
Therefore, the speed of the subsystem will, most of the time, be much smaller than $\| \haH_S \otimes \mathds{1} + \haH_{SB}\|_\infty$, which is the natural unit in which the speed of $\rho^S$ is to be measured \cite{0907.1267v1}.

\subsection{Fluctuations of the purity of the reduced state}
\label{sec:fluctutaionsofthepurityofthereducedstate}
The bound on the speed of the state of the subsystem we have discussed in the last section depends on the interaction Hamiltonian $\haH_{SB}$ and the local Hamiltonian $\haH_{S}$.
We are however primarily interested in understanding how the interaction with the environment leads to equilibration.
It is therefore instructive to consider a quantity that does not feel the local dynamics of the subsystem and instead is a good measure for the correlations with the environment.
Such a quantity is the purity of the subsystem \cite{Kimura07,Kimura09-1}
\begin{equation}
  p^S_t = p(\rho^S_t) = \Tr[{\rho^S_t}^2] .
\end{equation}
We can establish a bound on the time average of the rate of change of the purity
\begin{equation}
  \label{eq:rateofchangeofthepurity}
  \frac{dp^S_t}{dt} = \lim_{\delta t \to 0} \frac{p^S_t - p^S_{t + \delta t}}{\delta t}
\end{equation}
that depends only on the strength of the interaction Hamiltonian:
\begin{theorem}
  \label{theorem:puritychangesslowly}
  For every initial pure state $\psi_0$ of a composite system evolving under a Hamiltonian of the form of \eqref{eq:generalhamiltonianwithconstanttermabsorbed1} and with non-degenerate energy gaps, it holds that:
  \begin{equation}
    \expect{\left|\frac{dp^S_t}{dt}\right|}_t = \expect{\left|\frac{dp^B_t}{dt}\right|}_t \leq 2\,\|\haH_{SB}\|_\infty \sqrt{\frac{d_S^3}{d^{\mathrm{eff}}(\omega)}} .
  \end{equation}
  where $\omega = \expect{\psi_t}_t$.
\end{theorem}
\begin{proof}
  The first equality is trivial as the purity of the system and the purity of the bath are always identical if $\psi$ is pure. This follows from the Schmidt decomposition of the pure state \cite{nielsenm.a.c}.
  
  Note that 
  \begin{align}
    |\Tr[\rho^2] - \Tr[\sigma^2]| &= |\Tr[\rho^2 - \sigma^2]| \leq \Tr|\rho^2 - \sigma^2| \\
    &= 2 \Tr[\frac{\rho + \sigma}{2} (\rho - \sigma)] \\
    &\leq 2 \max_{0 \leq A \leq \mathds{1}}\Tr[A\,(\rho - \sigma)] \\
    &= 2 \tracedistance(\rho,\sigma) .
  \end{align}
  We can therefore bound the average rate of change of the purity by
  \begin{align}
    \left|\frac{dp^S_t}{dt}\right| &= \lim_{\delta t \to 0} \frac{|p^S_t - p^S_{t + \delta t}|}{\delta t} \\
    &\leq 2\,\lim_{\delta t \to 0} \frac{\tracedistance(\rho^S_t,\rho^S_{t + \delta t})}{\delta t} \\
    &\leq 2\,v_S(t) .
  \end{align}
  Inserting this into theorem~\ref{theorem:averagespeedisslow} gives
  \begin{equation}
    \expect{\left|\frac{dp^S_t}{dt}\right|}_t = \expect{\left|\frac{dp^B_t}{dt}\right|}_t \leq 2\,\| \haH_S \otimes \mathds{1} + \haH_{SB}\|_\infty \sqrt{\frac{d_S^3}{d^{\mathrm{eff}}(\omega)}} .
  \end{equation}
  Noting that $\haH_S$ does not influence the rate of change of the purity one obtains the desired result.
\end{proof}

The above theorem tells us that the average rate of change of the purity is small, in addition we can establish bounds telling us when the rate of change of the purity must be small during a particular evolution.
Form \eqref{eq:expectedvelocityintermsofcoefficeints} and \eqref{eq:coefficeintsforexpectedvelocity} we see that $d\rho^S_t/dt$ depends on $\haH_B$ only implicitly through the trajectory $\rho_t$. 
We therefore have:
\begin{equation}
  \rho^S_{t+\delta t} = \rho^S_t + \delta t\,\left( \iu\,[\rho^S_t,\haH_S] + \iu\,[\rho_t,\haH_{SB}] \right) + \mathcal{O}(\delta t^2)
\end{equation}
Inserting this into \eqref{eq:rateofchangeofthepurity} gives
\begin{align}
  \label{eq:rateofchangeofthepurityofthesubsystem}
  \frac{dp^S_t}{dt} &= \Tr[\rho^S_t\,2\,\iu\,\Tr_B[\rho_t,\haH_{SB}]] \\
  &= \Tr[\rho_t\,(2\,\iu\,\Tr_B[\rho_t,\haH_{SB}]\otimes\mathds{1})] .
\end{align}
The operator $2\,\iu\,\Tr_B[\rho_t,\haH_{SB}]$ is hermitian and traceless.
The more mixed $\rho^S_t$ is the more likely will it have overlap with both the eigenstates with positive and negative eigenvalues of this operator.
Therefore, the more mixed the subsystem is, the slower is the rate of change of its purity.
Rates near the maximal rate of change of $2\,\|\Tr_B[\rho_t,\haH_{SB}]\|_\infty \leq 2\,\|\haH_{SB}\|_\infty $ can only occur if $\rho_t$ is such that the subsystem state $\rho^S_t$ is relatively pure.
On the other hand if $\rho_t$ is a pure product state $p^S_t = 0$ and consequently $dp^S_t/dt = 0$ as $p^S_t$ is positive and differentiable.
Obviously, too little entanglement also leads to a slow rate of change of the purity.

The consequence of the interplay of these two counter acting influences is the subject of the following theorem:
\begin{theorem}
  \label{theorem:whenistherateofchangeofthepurityslow}
  For every initial state $\rho_0$ of a composite system evolving under a Hamiltonian of the form of \eqref{eq:generalhamiltonianwithconstanttermabsorbed1} and with non-degenerate energy gaps the absolute value of the rate of change of the purity is upper bounded by
  \begin{align}
    \left|\frac{dp^S_t}{dt}\right| &\leq 2\,\|\rho^S_t\|_\infty\,\sqrt{2\,I_{SB}(\rho_t)}\,\|\haH_{SB}\|_\infty\\
    &\leq 2\,p^S_t\,\sqrt{2\,I_{SB}(\rho_t)}\,\|\haH_{SB}\|_\infty ,
    \intertext{and if $\rho_t$ is pure this implies that}
    \left|\frac{dp^S_t}{dt}\right| &\leq 4\,p^S_t\,\sqrt{S(\rho^S_t)}\,\|\haH_{SB}\|_\infty ,
  \end{align}
  where $\rho^{\mathrm{cor}}_t = \rho_t - \rho^S_t \otimes \rho^B_t$, $S(\rho^S_t)$ is the von~Neumann entropy of the reduced state and $I_{SB}(\rho_t)$ is the quantum mutual information between the subsystem and the bath.
\end{theorem}
\begin{proof}
  In \cite{Kimura07} it is shown that instead of \eqref{eq:rateofchangeofthepurityofthesubsystem} we may write the rate of change of the purity as
  \begin{equation}
    \frac{dp^S_t}{dt} = \Tr[\rho^S_t\,2\,\iu\,\Tr_B[\rho^{\mathrm{cor}}_t,\haH_{SB}]]
  \end{equation}
  where $\rho^{\mathrm{cor}}_t = \rho_t - \rho^S_t \otimes \rho^B_t$ is the \emph{correlation operator} \cite{Kimura07} which satisfies
  \begin{equation}
    \|\rho^{\mathrm{cor}}_t\|_1 \leq \sqrt{2\,I_{SB}(\rho_t)} 
  \end{equation}
  and $I_{SB}(\rho_t)$ is the quantum mutual information.
  Using this one can easily see that
  \begin{align}
    \left|\frac{dp^S_t}{dt}\right| &\leq 2\,\|\rho^S_t\|_\infty\,\|[\rho^{\mathrm{cor}}_t,\haH_{SB}]\|_1 \\
    &\leq 2\,\|\rho^S_t\|_\infty\,\|\haH_{SB}\|_\infty\,\|\rho^{\mathrm{cor}}_t\|_1 \\
    &\leq 2\,\|\rho^S_t\|_\infty\,\|\haH_{SB}\|_\infty\,\sqrt{2 I_{SB}(\rho_t)}\\
    &\leq 2\,p^S_t\,\sqrt{2 I_{SB}(\rho_t)}\,\|\haH_{SB}\|_\infty ,
  \end{align}
  where in the last inequality we have used the fact that $\|\rho^S_t\|_\infty \leq \|\rho^S_t\|_2 = p^S_t$.
  For pure states $\rho_t$ the quantum mutual information reduces to
  \begin{equation}
      I_{SB}(\rho_t) = S(\rho^S_t) + S(\rho^B_t) - S(\rho_t) = 2\,S(\rho^S_t) .
  \end{equation}
\end{proof}

\subsection{Equilibration time}
\label{sec:equilibrationtime}
Only if equilibration happens on reasonable time scales the mechanism of equilibration presented herein gives a satisfactory explanation for the irreversible behavior of our every days world.

It is obvious that when the Hamiltonian is multiplied by a positive constant factor the dynamics of the system speeds up or slows down by exactly this factor.
Similarly one might expect that a rescaling of the interaction part of the Hamiltonian will increase or decrease the equilibration time.
But, other properties, like the relative orientation of the eigenbasis of the subsystem Hamiltonian and that of the interaction Hamiltonian, its interaction range, the spacial extend of the system, as well as the fine structure of the spectrum also have crucial influence on these timescales.
Due to the generality of the approach pursued herein we cannot say much about the timescales on which equilibration happens, although we will give some bounds below.
This has provoked well justified criticism \cite{Cho09}.

However, In specific models it is possible to calculate equilibration times explicitly and it turns out that they have reasonable values.
This was shown in some analytical works \cite{Wang08,Cramer08,Cramer09} as well as in numerical studies \cite{Borowski03,0708.1324v1,Barthel08,0904.1501v1,Devi09,Gemmer09}.

Some simple estimates can be made even without specifying a model in detail:
Assume that the initial state $\psi_0$ was drawn from the subspace of energy eigenstates with eigenvalues in the interval $[E,E + \Delta E]$.
Then
\begin{align}
  v_S(t) &= \lim_{\delta t \to 0} \frac{\tracedistance(\rho^S_t,\rho^S_{t+\delta t})}{\delta t} \\
  &= \frac{1}{2} \left\| \frac{d\rho^S_t}{dt} \right\|_1 \leq \frac{1}{2} \left\| \frac{d\rho_t}{dt} \right\|_1 \\
  &= \frac{1}{2} \left| [\haH,\rho_t]\right|_1 \leq \Delta E
\end{align}
If the system starts in a non-equilibrium state the initial distance $\tracedistance(\rho^S_0,\omega^S)$ from the time averaged, equilibrium state $\omega^S$ can be expected to be of order 1.
So that even if the subsystem state immediately starts to head towards $\omega$ with maximal speed it will take at least a time span of the order of magnitude of the Heisenberg time  
\begin{equation}
  T \approx \frac{1}{\Delta E}
\end{equation}
until the equilibrium state is reached.
That equilibration can indeed happen on timescales that are roughly of the order of magnitude of $T$ can be seen from the numerical simulations presented in \cite{Borowski03}.

Another non trivial bound on the equilibration time can be obtained from theorem~\ref{theorem:whenistherateofchangeofthepurityslow}.
For pure joint system states we found the following bound on the rate of change of the purity:
\begin{align}
  \label{eq:purityspeedbounddgl}
  \left|\frac{dp^S_t}{dt}\right| &\leq 4\,p^S_t\,\sqrt{S(\rho^S_t)}\,\|\haH_{SB}\|_\infty \\
  &\leq 4\,p^S_t\,\sqrt{\log(d_S)}\,\|\haH_{SB}\|_\infty 
\end{align}
Now, assume that the initial state is a pure product state so that $p^S_0 = 1$ and that the equilibrium state has purity $p^S_{\mathrm{eq}}$.
By integrating the differential equation for the purity \eqref{eq:purityspeedbounddgl} one finds that the time $T$ to reach the equilibrium purity $p^S_{\mathrm{eq}}$ is bounded by 
\begin{equation}
  \label{eq:boundontheequilibrationtmeforfixedtargetpurity}
  T \geq \frac{\log(\frac{1}{p^S_{\mathrm{eq}}})}{4\,\sqrt{\log(d_S)}\,\|\haH_{SB}\|_\infty} .
\end{equation}
Equilibration in a shorter time is impossible, even if the evolution is such that the purity decreases with the maximal possible rate.
If the equilibrium state is the maximally mixed state $p^S_{\mathrm{eq}} = 1/d_S$ the minimum time until equilibration is
\begin{equation}
  T \geq \frac{\sqrt{\log(d_S)}}{4\,\|\haH_{SB}\|_\infty} .
\end{equation}

We find that equilibration to a state with fixed purity can happen the faster the larger the system is.
The minimal time decreases like $1/\sqrt{d_S}$. 
In contrast, equilibration to the completely mixed state takes longer the larger the system is.
Here the minimal time increases with $\sqrt{d_S}$.
In both cases most time is spend during the final approach, as according to \eqref{eq:purityspeedbounddgl} the rate of change gets slower the lower the purity is.
The more relevant time scale, even for equilibration towards the completely mixed state therefore is \eqref{eq:boundontheequilibrationtmeforfixedtargetpurity}, as there will be some value of the purity, independent of $d_S$, from which on the state will be practically indistinguishable from the completely mixed state.

\section{Equilibration and einselection}
\label{sec:equilibrationandeinselection}
The term \emph{einselection}, which stands for \emph{environment-induced super selection}, is due to Zurek \cite{PhysRevD.26.18,RevModPhys.75.715}.
Einselection is known to occur in situations where the Hamiltonian of the composite system leaves a certain orthonormal basis of the subsystem, spanned by so called \emph{pointer states} $\ket{p}$, invariant \cite{Hornberger09}.
If this is the case, the Hamiltonian and the time evolution operator have the form
\begin{align}
  \label{eq:einselectionhamiltonian}
  \haH &= \sum_p \ketbra{p}{p} \otimes \haH^{(p)} \\
  U_t &= \sum_p \ketbra{p}{p} \otimes U^{(p)}_t ,
\end{align}
where $U^{(p)}_t = \ee^{-\iu\,\haH^{(p)}\,t}$ and the $\haH^{(p)}$ are arbitrary hermitian matrices.
One finds that the subsystem state of an initial product state of the form $\rho_0 = \rho^S_0 \otimes \psi^B_0$, where the state of the bath can be assumed to be pure without loss of generality, evolves into
\begin{equation}
  \rho^S_t = \sum_{pp'} \ketbra{p}{p}\rho^S_0\ketbra{p'}{p'}\,\bra{\psi^B_0}{U^{(p')}_t}^\dagger\,U^{(p)}_t\ket{\psi^B_0}
\end{equation}
Under the evolution induced by such a Hamiltonian the diagonal entries of $\rho^S_0$, when expressed in the pointer basis, remain unchanged while the off-diagonal entries are suppressed by a factor of $\bra{\psi^B_0}{U^{(p')}_t}^\dagger\,U^{(p)}_t\ket{\psi^B_0} \leq 1$.
The actual time development of the $\bra{\psi^B_0}{U^{(p')}_t}^\dagger\,U^{(p)}_t\ket{\psi^B_0}$ depends on the explicit model under consideration, but for many models they have been found to decrease rapidly over short time scales \cite{zeh96,Breuer02,Hornberger09,RevModPhys.75.715,PhysRevD.26.18}.
If some of the $\haH^{(p)}$ lead to an identical time development for the chosen initial bath state there exist subspaces of $\hiH_S$ in which coherence is preserved and in which quantum mechanical superpositions survive the interaction with the environment. 

Note that, the diagonal entries, which survive the decoherence, are completely determined by $\rho^S_0$ and do not depend on the initial state of the bath $\psi^B_0$ at all.
The direct opposite situation is the thermodynamic case where the final state is completely determined by the properties of the bath.
Most realistic situations surely lie between these two extremes.

Using the results discussed in section~\ref{sec:speedoffluctuationsaroudequilibrium} it is possible to get rid of the quite limiting assumption on the form of the Hamiltonian and to shown that einselection is a more general phenomenon.
Besides the usual assumption of non-degenerate energy gaps and $d_B \gg d_S$, to ensure a large average effective dimension for almost all pure initial states, we only need to assume that the interaction Hamiltonian $\haH_{SB}$ is weak.
Note however, that in this more general setting, we can currently not say much about time scale on which the decoherence happens (see section~\ref{sec:equilibrationtime} for more details). 
The following discussion was earlier published by the author in \cite{Gogolin09-1}.

We have seen in section~\ref{sec:speedoffluctuationsaroudequilibrium} that for a system with a Hamiltonian of the form
\begin{equation}
  \label{eq:generalhamiltonianwithconstanttermabsorbed2}
  \haH = \haH_0 + \haH_S \otimes \mathds{1} + \mathds{1} \otimes \haH_B + \haH_{SB} ,
\end{equation}
where $\haH_0$ is proportional to the identity and $\haH_S$, $\haH_B$ and $\haH_{SB}$ are traceless, the velocity of the subsystem state $\rho^S_t$ is 
\begin{equation}
  \frac{d\rho^S_t}{dt} = \sum_{k=1}^{d_S^2} c_k(t)\,e_k
\end{equation}
where the $d_S^2$ operators $e_k$ form a hermitian orthonormal basis for $\hiH_S$ such that $\Tr[e_k\,e_l] = \delta_{kl}$ and
\begin{equation}
  c_k(t) = \Tr\big[\rho(t)\,\iu\,[\haH_S \otimes \mathds{1} + \haH_{SB},e_k\otimes \mathds{1}]\big] .
\end{equation}
The velocity depends on $\haH_B$ only implicitly through the trajectory $\rho_t$, but for an arbitrary fixed state $\rho$ the velocity is solely determined by $\haH_S$ and $\haH_{SB}$:
\begin{equation}
  \label{eq:subsystemstatevelocity}
  \frac{d\rho^S}{dt} = \iu\,[\rho^S,\haH_S] + \iu\,\Tr_B[\rho,\haH_{SB}]
\end{equation}
Now if $\haH_{SB}$ is much weaker than $\haH_S$, \eqref{eq:subsystemstatevelocity} is dominated by the first term.
Consequently, the system can only become slow when $[\rho^S,\haH_S]$ is small.

To see when this happens we first establish a general lower bound on the norm of commutators between states and arbitrary hermitian matrices:
\begin{lemma}
  \label{lemma:lowerboundonnormsofcommuators}
  Let $\rho$ be a normalized state and $A$ a hermitian observable with eigenvalues $a_k$ and eigenvectors $\ket{a_k}$, then
  \begin{align}
    \label{eq:commutatorbound}
    \| [\rho,A] \|_1 = \| \iu\,[\rho,A] \|_1 &\geq 2\max_{\{(k,l)\}} \sum_{(k,l)} |a_k - a_l|\,|\rho_{kl}|\\
    &\geq 2\max_{kl} |a_k - a_l|\,|\rho_{kl}| .
  \end{align}
  where the maximization in \eqref{eq:commutatorbound} is performed over all decompositions of the index set $\{1,\dots,d_S\}$ into non-overlapping pairs $(k,l)$ over which the sum is performed and $\rho_{kl} = \bra{a_k}\rho\ket{a_l}$.
\end{lemma}
\begin{proof}
  The equality is trivial.
  For all traceless, hermitian, bounded operators $B$ on some finite dimensional Hilbert space $\hiH$ it holds that \cite{nielsenm.a.c}
  \begin{equation}
    \|B\|_1 = 2\,\max_{\Pi \in \mathcal{P}(\hiH)} \Tr[\Pi\,B] ,
  \end{equation}
  where $\mathcal{P}(\hiH)$ is the set of all projectors on $\hiH$ and the maximum is obtained when $\Pi$ is the projector onto the positive subspace of $B$.
  By expanding $\rho$ in the eigenbasis of $A$, using the above equality for $B=[\rho,A]$ and considering all sums of mutually orthogonal rank one projectors $\Pi_{kl}$ of the form
  \begin{equation}
    \Pi_{kl} = \ketbra{\pi_{kl}}{\pi_{kl}} \qquad \ket{\pi_{kl}} = \frac{1}{\sqrt{2}}(\ket{a_k} + \ee^{\iu \phi_{kl}} \ket{a_l}) ,
  \end{equation}
  where $\phi_{kl}$ are phase factors, one easily verifies \eqref{eq:commutatorbound}.
  The second inequality is trivial.
\end{proof}
Using the above lemma we can now proof the following theorem:
\begin{theorem}
  \label{theorem:slowstatesmustdecohere}
  Consider a physical system evolving under a Hamiltonian of the form given in \eqref{eq:generalhamiltonianwithconstanttermabsorbed1} and with non-degenerate energy gaps.
  All reduced states $\rho^S$ satisfy
  \begin{align}
    \|\haH_{SB}\|_{\infty} + \frac{1}{2}\left\|\frac{d\rho^S}{dt}\right\|_1 &\geq \max_{\{(k,l)\}} \sum_{(k,l)} |E^S_k - E^S_l|\,|\rho^S_{kl}|\\
    &\geq \max_{kl} |E^S_k - E^S_l|\,|\rho^S_{kl}| ,
  \end{align}
  where $\rho^S_{kl} = \bra{E^S_k} \rho^S \ket{E^S_l}$ and $E^S_k$ and $\ket{E^S_k}$ are the eigenvalues and eigenstates of $\haH_S$.
\end{theorem}
\begin{proof}
  Using the inverse triangle inequality and \eqref{eq:subsystemstatevelocity} we see that
  \begin{equation}
    | \| \iu\,[\rho^S,\haH_S] \|_1 - \| \iu\,\Tr_B[\rho,\haH_{SB}] \|_1 | \leq \left\|\frac{d\rho^S}{dt}\right\|_1 .
  \end{equation}
  For $\|d\rho^S/dt\|_1$ to become small the norms of the two commutators must be approximately equal.
  Applying lemma~\ref{lemma:lowerboundonnormsofcommuators} to the norm of the first commutator yields:
  \begin{equation}
    \| \iu\,[\rho^S,\haH_S] \|_1 \geq 2\,\max_{k \neq l} |E^S_k - E^S_l|\,|\rho^S_{kl}|
  \end{equation}
  The norm of the second commutator can be upper bounded, using the well-known fact that the trace norm of traceless, hermitian matrices is non-increasing under completely positive, hermitian, trace-non-increasing maps \cite{nielsenm.a.c} as follows:
  \begin{equation}
    \| \iu\,\Tr_B[\rho,\haH_{SB}] \|_1 \leq \| [\rho,\haH_{SB}] \|_1 \leq 2\,\|\haH_{SB}\|_{\infty}
  \end{equation}
  This completes the proof.
\end{proof}

\begin{figure}[p]
  \centering\footnotesize
  \begin{equation*}
    \frac{d\rho_t^S}{dt} = \iu\,\Tr_B[\psi_t,\haH_S \otimes \mathds{1} + \haH_{SB}] = {\color{structure}\iu\,[\rho^S_t,\haH_S]}  + {\color{red}\iu\, \Tr_B[\psi_t,\haH_{SB}]}      
  \end{equation*}
  \subfloat[]{%
    \inputTikZ{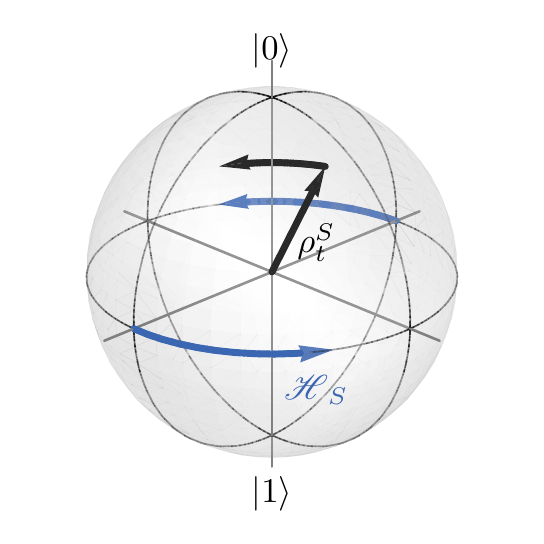}
%     \begin{tikzpicture}[scale=0.75]
%       \node at (0,0) {\includegraphics{speedunderweakinteraction1.pdf}};
%       \node[structure] at (0.6,-1.6) {$\haH_{S}$};
%       \node[black] at (0.6,0.4) {$\rho^S_t$};
%       \node at (0, 3) {$\ket{0}$};
%       \node at (0,-3) {$\ket{1}$};
%     \end{tikzpicture}
    }\hfill\subfloat[]{%
      \inputTikZ{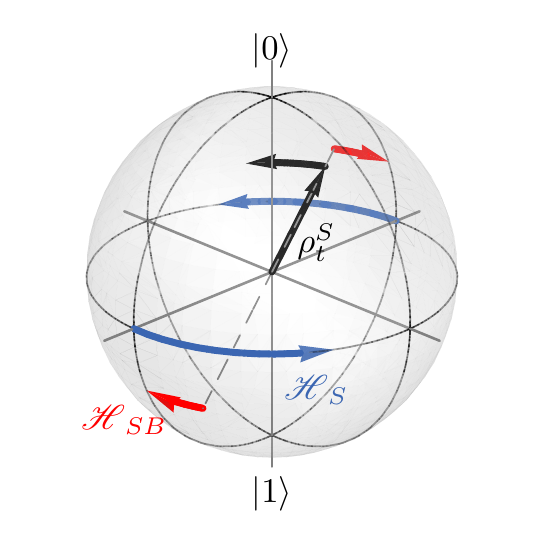}
%     \begin{tikzpicture}[scale=0.75]
%       \node at (0,0) {\includegraphics{speedunderweakinteraction2.pdf}};
%       \node[structure] at (0.6,-1.6) {$\haH_{S}$};
%       \node[black] at (0.6,0.4) {$\rho^S_t$};
%       \node[red] at (-2,-2) {$\haH_{SB}$};
%       \node at (0, 3) {$\ket{0}$};
%       \node at (0,-3) {$\ket{1}$};
%     \end{tikzpicture}
    }\\
    \subfloat[]{%
      \inputTikZ{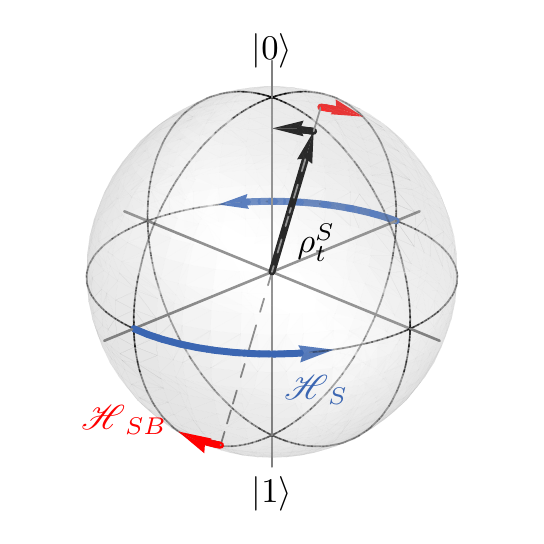}
%     \begin{tikzpicture}[scale=0.75]
%       \node at (0,0) {\includegraphics{speedunderweakinteraction3.pdf}};
%       \node[structure] at (0.6,-1.6) {$\haH_{S}$};
%       \node[black] at (0.6,0.4) {$\rho^S_t$};
%       \node[red] at (-2,-2) {$\haH_{SB}$};
%       \node at (0, 3) {$\ket{0}$};
%       \node at (0,-3) {$\ket{1}$};
%     \end{tikzpicture}
  }\hfill\subfloat[]{%
    \inputTikZ{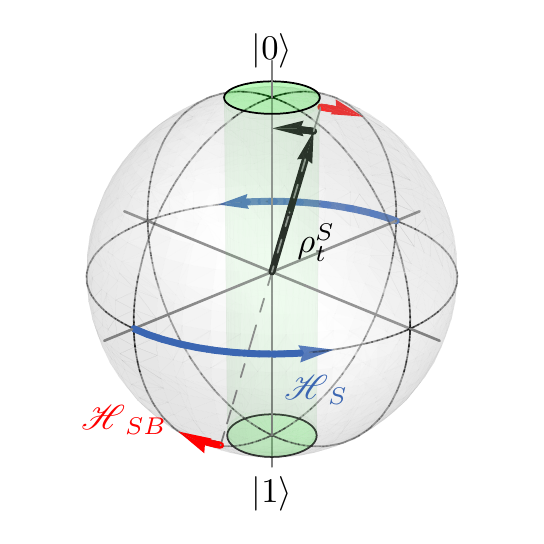}
%     \begin{tikzpicture}[scale=0.75]
%       \node at (0,0) {\includegraphics{speedunderweakinteraction4.pdf}};
%       \node[structure] at (0.6,-1.6) {$\haH_{S}$};
%       \node[black] at (0.6,0.4) {$\rho^S_t$};
%       \node[red] at (-2,-2) {$\haH_{SB}$};
%       \node at (0, 3) {$\ket{0}$};
%       \node at (0,-3) {$\ket{1}$};
%     \end{tikzpicture}
  }
  \caption{
    Imagine the small subsystem is a single spin $1/2$ particle and $\haH_S$ has eigenvectors $\ket{0}$ and $\ket{1}$.
    If no interaction with the bath were present the state of the subsystem would rotate around the $z$-axis with a constant speed (a).
    The interaction Hamiltonian $\haH_{SB}$ gives an additional contribution to the velocity  of the reduced state.
    As the interaction is assumed to be weak it can in general not significantly slow down the rotation of the reduced state due to the local Hamiltonian, even if it its contribution to the velocity points in the directly opposite direction (b).
    Only if $\rho^S_t$ is close to an eigenstate of $\haH_S$, such that $[\rho^S_t,\haH_S]$ becomes small, the speed of the reduced state can become slow (c).
    The interplay of the two Hamiltonians defines a set of \emph{slow states} (green) (d).
    As we know from theorem~\ref{theorem:averagespeedisslow} that the speed of the subsystem state is almost always slow $\rho^S_t$ must spend most of the time in this allowed region.
  }
  \label{fig:slowstates}
\end{figure}

The assertion of theorem~\ref{theorem:slowstatesmustdecohere} is almost intuitively clear, but combined with theorem~\ref{theorem:averagespeedisslow} it allows to draw the following powerful conclusion:
Whenever $d^{\mathrm{eff}}(\omega)$ is large the subsystem is slow most of the time and if this is the case coherent superpositions of eigenstates of $\haH_S$ with eigenvalue differences that are much larger than $\|\haH_{SB}\|_\infty$ may not contribute significantly to the state of the subsystem.
That is, the corresponding off-diagonal elements of the reduced state $\rho^S_t$ in the $\haH_S$ eigenbasis must be small.
A similar behavior was observed for a specific model in \cite{PhysRevLett.82}.
Without using any approximations we have shown that coherence can only be retained between eigenstates of $\haH_S$ whose energies differ by less than $\|\haH_{SB}\|_\infty$.
This statement remains meaningful even when the subsystem is large and its energy spectrum thus very dense.
Theorem~\ref{theorem:slowstatesmustdecohere} then still implies that coherent superpositions of eigenstates with far apart energies (sometimes called Sch\"{o}dinger cat states) must decohere.
If the subsystem is small and the interaction Hamiltonian weak compared to the energy gaps of the subsystem Hamiltonian it implies an even stronger statement.
The state of the subsystem must then, most of the time, be approximately diagonal in the eigenbasis of $\haH_S$.

The consequences of theorem~\ref{theorem:slowstatesmustdecohere} are twofold:

1) It proofs the existence of a natural decoherence mechanism in weakly interacting systems that causes decoherence in the local energy eigenbasis.
This effect is indeed observed in many situations where the local Hamiltonian is much stronger than the interaction.
A well-known example are electronic excitations of gases at moderate temperature.
The energy gaps between the ground state and the first few excited states are typically much larger than the thermal energy.
The dynamics of such systems is successfully described using transition rates between energy eigenstates.
Ultimately theorem~\ref{theorem:slowstatesmustdecohere} explains why this is eligible.

2) It can be seen as an intermediate step of proving relaxation to the Gibbs state.
It goes beyond the results of \cite{Linden09} and sheds some light on the roll of the weakness of the interaction in thermalization.
This is interesting as the derivations of the canonical state given so fare \cite{slloydthesis,tasaki98,Gemmer02,Goldstein06,Cho09} either need to make very special assumptions, or at least partly depend on heuristic arguments.
See section~\ref{sec:thecanonicalensemble} and \ref{sec:conditionsonthehamiltonian} for a more detailed analysis of these results and a more elaborate explanation of this criticism.

\section{Initial state independence and the Second~Law}
\label{sec:initialstateindependence}
The final state of a small subsystem coupled to a large heat bath is typically independent of its initial state and is completely determined by some macroscopic properties of the bath.
This key feature of thermodynamic systems is expressed in the Second~Law of Thermodynamics.
The obvious example is the canonical Boltzmann state which is completely specified by giving the local Hamiltonian and the temperature of the bath.
One would therefore like to know under which conditions the equilibrium state of the subsystem of a large quantum mechanical system is independent from the initial state and depends only on the subspace from which the bath state was drawn.

Showing initial state independence is a prerequisite for thermalization and an important step towards a derivation of a Quantum Mechanical replacement of the Second~Law.
Unfortunately, it seems to be a quite difficult task to identify reasonable conditions under which initial state independence holds and there certainly are situations where it is violated.

\subsection{Conditions on the Hamiltonian}
\label{sec:conditionsonthehamiltonian}
From the form of the time averaged, equilibrium state it is obvious that the Hamiltonian, and in particular the form of the marginals of the populated energy eigenstates, is decisive for whether the equilibrium state depends on the initial state of the subsystem or not.

In theorem~3 in \cite{Linden09} the following bound on the trace distance of the time averaged system state $\omega^S = \expect{\Tr_B[\psi_t]}_t$, that belongs to a random initial pure state $\rho_0 \in \mathcal{P}_1(\hiH_R)$ chosen from a subspace $\hiH_R \subset \hiH = \hiH_S \otimes \hiH_B$, form the reduced microcanonical state $\rho_{\mathrm{mc}}^S = \Tr_B[\rho_{\mathrm{mc}}]$ is established:
\begin{equation}
  \expect{\tracedistance(\omega^S,\rho_{\mathrm{mc}}^S)}_{\psi_0} \leq \sqrt{\frac{d_S\,\delta}{4\,d_R}}
\end{equation}
Thereby $\expect{\cdot}_{\psi_0}$ is the average over random pure initial states,
\begin{equation}
  \delta = \sum_k \bra{E_k} \frac{\Pi_R}{d_R} \ket{E_k} \Tr_S[(\Tr_B \ketbra{E_k}{E_k})^2 ] \leq 1 ,
\end{equation}
and $\Pi_R$ is the projector onto $\hiH_R$.
This theorem is used to argue that if the state of the bath is fixed, i.e. $\hiH_R = \hiH_S \otimes \ket{\varphi}_B$ and thereby $d_R = d_S$, the time averaged state of the subsystem is independent of its initial state if the energy eigenstate are highly entangled, i.e if
\begin{equation}
  \Tr_B \ketbra{E_k}{E_k} \approx \frac{1}{d_S} \mathds{1}_{d_S \times d_S} .
\end{equation}
However, even under this assumption one has:
\begin{align}
  \delta &\approx \sum_k \bra{E_k}\frac{\Pi_R}{d_R}\ket{E_k} \Tr_S[(\frac{1}{d_S} \mathds{1}_{d_S \times d_S})^2 ]\\
  &= \frac{1}{d_R\,d_S} \sum_k \bra{E_k}\Pi_R\ket{E_k} = \frac{1}{d_S}
\end{align}
Consequently the best estimate one can gain by applying the above bound is:
\begin{equation}
  \expect{\tracedistance(\omega^S_\psi,\rho_{\mathrm{mc}}^S)}_\psi \leq \sqrt{\frac{1}{4\,d_S}}
\end{equation}
This result is quite counter intuitive.
One would expect that the time averaged equilibrium state $\omega_S$ depends less on the initial state of the subsystem the smaller it is compared to the bath.
In contrast to this the above bound gets tighter the larger $d_S$ is and does not depend on $d_B$ at all.
One can therefore anticipate that the above bound is not tight for small subsystems and large baths.

And indeed, if all energy eigenstates in some subspace have similar marginals, and in particular, if they are all highly entangled, it is possible to show initial state independence whenever $d_S \ll d^{\mathrm{eff}}(\omega)$:
\begin{theorem}
  \label{theorem:initialstateindependence}
  Let all energy eigenstates that span some subspace $\hiH_R \subseteq \hiH$ of the total systems Hilbert space have similar marginals in the sense that
  \begin{equation}
    \max_{\ket{E_k},\ket{E_l} \in \hiH_R} \tracedistance(\Tr_B\ketbra{E_k}{E_k},\Tr_B\ketbra{E_l}{E_l}) = \delta ,
  \end{equation}
  or be highly entangled in the sense that
  \begin{equation}
    \max_{\ket{E_k} \in \hiH_R} 2\,\tracedistance(\Tr_B\ketbra{E_k}{E_k},\frac{\mathds{1}}{d_S}) = \delta .
  \end{equation}
  Then the time averaged distance of the marginals $\rho^S_t$ and $\sigma^S_t$ of any two pure initial states from $\hiH_R$ evolving under a Hamiltonian with non-degenerate energy gaps is upper bounded by 
  \begin{equation}
    \expect{\tracedistance(\rho^S_t,\sigma^S_t)}_t \leq  \frac{1}{2} \sqrt{\frac{d_S}{d^\mathrm{eff}(\expect{\rho^B_t}_t)}} + \frac{1}{2} \sqrt{\frac{d_S}{d^\mathrm{eff}(\expect{\sigma^B_t}_t)}} + \delta.
  \end{equation}
\end{theorem}
\begin{proof}
  Using the triangle inequality twice we see that
  \begin{equation}
    \label{eq:initialstateindepenceconditionfromtriangularinequality}
    \begin{split}
      \expect{\tracedistance(\rho^S_t,\sigma^S_t)}_t \leq \expect{\tracedistance(\rho^S_t,\expect{\rho^S_\tau}_\tau)}_t + \expect{\tracedistance(\sigma^S_t,\expect{\sigma^S_\tau}_\tau)}_t + \tracedistance(\expect{\rho^S_t}_t,\expect{\sigma^S_t}_t) .
    \end{split}
  \end{equation}
  The terms $\expect{\tracedistance(\rho^S_t,\expect{\rho^S_\tau}_\tau)}_t$ and $\expect{\tracedistance(\sigma^S_t,\expect{\sigma^S_\tau}_\tau)}_t$ can both be bounded using theorem~\ref{theorem:distancefromtimeaverage} and
  \begin{equation}
    \label{eq:badestimateofthetrace}
    \tracedistance(\expect{\rho^S_t}_t,\expect{\sigma^S_t}_t) \leq \delta .
  \end{equation}
\end{proof}

On the first sight theorem~\ref{theorem:initialstateindependence} seems to be a quite nice result.
The assumption of similar reduced states allows for some dependence of the final state of the subsystem on the macroscopic features of the bath.
For example, assume that the marginals of the energy eigenstates are all close to the canonical state for the temperature associated with their respective energy.
All marginals of the eigenstates that belong to some energy interval would then be close to the correct canonical state and we would recover the situation of equilibration towards a Boltzmann distribution known from classical Statistical Mechanics. 

Unfortunately, if the macroscopic properties are not primarily determined by the subsystem, and this is exactly the situation we are interested in, it seems to be unreasonable to assume that most energy eigenstates from some energy subspace have similar marginals.
This assumption is called the \emph{eigenstate thermalization hypothesis} and was first suggested in \cite{PhysRevE.50.88} (see also \cite{0708.1324v1}) to explain thermalization.
Although it was shown in \cite{0708.1324v1} that the expectation values in the energy eigenstates of some reduced observables of an example system are close to a continuous function of energy the author does not find this assumption very convincing.
In the very weak coupling limit, where we would like to recover the Boltzmann distribution, the energy eigenstates are usually assumed to be close to product \cite{Goldstein06} and their marginals thus are not at all similar.
Changing the subsystem part of the initial state will have a significant impact on which eigenstates have a non vanishing overlap $|c_k|^2 = |\braket{E_k}{\psi_0}|^2$ with the initial state.
Obviously only those $\ket{E_k}$ whose system marginal are similar to the system part of the initial state, i.e. the states with $\Tr_B[\ketbra{E_k}{E_k}] \approx \psi^S_0$, will be populated.
Thus, the initial state of the system will have a non negligible impact on how the equilibrium state 
\begin{equation}
  \omega^S = \expect{\rho^S_t}_t = \sum_k |c_k|^2 \Tr_B[\ketbra{E_k}{E_k}] ,
\end{equation}
will look like if the energy eigenstates $\ket{E_k}$ are close to product.
This is true even for initial states with a high average effective dimension.
Thus we can get equilibration without initial state independence and thus without thermalization.\footnote{A work that will elaborate more on this point is currently in preparation.}

What do we learn from that?
To have a chance of proving initial state independence, one at least needs \emph{some} entanglement in the energy eigenstates and therefore a coupling Hamiltonian $\haH_{SB}$ which is in strength at least comparable with the energy gaps of the non-interacting Hamiltonian $\haH_S \otimes \mathds{1} + \mathds{1} \otimes \haH_B$. 

This condition is not to be confused with the assumption on the gaps of $\haH_S$ we have worked with in section~\ref{sec:equilibrationandeinselection}.
It is well possible that $\|\haH_{SB}\|_\infty$ is small compared to the gaps $\Delta(\haH_S)$ of the subsystem Hamiltonian while at the same time large compared to the gaps $\Delta(\haH_S \otimes \mathds{1} + \mathds{1} \otimes \haH_B)$ of the non-interacting part of the Hamiltonian:
\begin{equation}
  \Delta(\haH_S \otimes \mathds{1} + \mathds{1} \otimes \haH_B) \ll \|\haH_{SB}\|_\infty \ll \Delta(\haH_S)
\end{equation}
This can be expected to be the natural situation in thermodynamically large systems, as the density of energy states typically increases exponentially with the size of the system. 

If $\|\haH_{SB}\|_\infty \ll \Delta(\haH_S \otimes \mathds{1} + \mathds{1} \otimes \haH_B)$ the equilibrium state can be expected to not be robust against unitary transformations of the subsystem part of the initial state.
Works claiming to derive the canonical ensemble under this assumption, or equivalently under the assumption that the energy eigenstates are close to product, and works that do not explicitly exclude the case $\|\haH_{SB}\|_\infty \ll \Delta(\haH_S \otimes \mathds{1} + \mathds{1} \otimes \haH_B)$ should therefore be considered with a healthy amount of mistrust (compare \cite{Goldstein06,PhysRev.114.94,Cho09}).

\subsection{Highly entangled eigenstates and random Hamiltonians}
\label{sec:highlyentangledeigenstatesaregeneric}
As we have seen in the last section, highly entangled energy eigenstates are sufficient for initial state independence.
However, this intuitively seems to be a very special property presumably not found in most realistic systems.
It is however possible to show that almost all random Hamiltonians actually have highly entangled eigenstates.
Where \emph{almost all} has a mathematically precise and well defined meaning.

The argument is based on the fact that random pure states are highly entangled with very high probability:
\begin{lemma}
  \label{lemma:randomstatesarehighlyentangled}
  Given the Hilbert space $\hiH = \hiH_S \otimes \hiH_B$ of a composite system with $d_S \ll d_B$.
  The reduced state $\rho^S = \Tr_B \psi$ of a random pure state $\psi \in \mathcal{P}_1(\hiH)$ is with very high probability highly entangled in the sense that
  \begin{equation}
    \label{eq:randomstateoneoverepsilonentanglementprobability}
    \probability\left\{ \tracedistance( \rho^S, \frac{\mathds{1}}{d_S}) \geq \epsilon \right\} \leq 2 \left(\frac{10\,d_S}{\epsilon}\right)^{2\,d_S} \ee^{-C\,d_B\,\epsilon^2} ,
  \end{equation}
  where $C=(14\,\ln(2))^{-1}$.
\end{lemma}
\begin{proof}
  The lemma is a direct corollary of lemma III.4 in \cite{Hayden06}, which establishes a bound on the probability that one of the eigenvalues $\lambda_i$ of the reduced state $\rho^S = \Tr_B \psi $ of a random pure state $\psi \in \mathcal{P}_1(\hiH)$ differs from $1/d_S$ by more than $\epsilon/d_S$, namely:
  \begin{equation}
    \probability\left\{\exists\ i :\ | \lambda_i - \frac{1}{d_S} | \geq \frac{\epsilon}{d_S} \right\} \leq 2 \left(\frac{10\,d_S}{\epsilon}\right)^{2\,d_S} \ee^{-C\,d_B\,\epsilon^2}
  \end{equation}
  where $C=(14\,\ln(2))^{-1}$.
  If non of the eigenvalues of $\Tr_B \psi $ differs from $1/d_s$ by more than $\epsilon/d_S$ then $\tracedistance(\psi^S,\mathds{1}/d_S) \leq \epsilon$.
\end{proof}

Before we can proceed we must specify what we mean by \emph{random Hamiltonian}.
A random Hamiltonian is a hermitian matrix whose eigenbasis was chosen according to the unitary invariant Haar measure.\footnote{After finishing this section of the present work related considerations, but with a quite different intention, were published in \cite{Goldstein09}.}
We do not put any special restrictions on the eigenvalues of the Hamiltonian except that we assume that the energy gaps are non-degenerate and the spectrum bounded.
A uniform random orthogonal basis for $\hiH = \hiH_S \otimes \hiH_B$, which we think of as the eigenbasis of such a random Hamiltonian, can be constructed by applying the same random unitary transformation, chosen according to the Haar measure of the unitary group, on every element of an arbitrary initial basis.
The eigenvectors $\ket{E_k}$ of such a basis thus each look exactly as if they where random vectors in $\hiH$.
Using the union bound and lemma~\ref{lemma:randomstatesarehighlyentangled} we therefore find that:\footnote{This simple proof of lemma~\ref{lemma:entangledeigenstasaregeneric} was suggested by Andreas Winter after the author had established a slightly weaker statement with a much more involved proof.}
\begin{lemma}
  \label{lemma:entangledeigenstasaregeneric}
  All eigenstates $\ket{E_k}$ of a random Hamiltonian on $\hiH = \hiH_S \otimes \hiH_B$ with $d_S \ll d_S$ are with high probability close to maximally entangled in the sense that 
  \begin{equation}
    \probability\left\{ \exists k: \tracedistance( \Tr_B \ketbra{E_k}{E_k} , \frac{\mathds{1}}{d_S}) \geq \epsilon \right\} \leq 2\,d \left(\frac{10\,d_S}{\epsilon}\right)^{2\,d_S} \ee^{-C\,d_B\,\epsilon^2} ,
  \end{equation}
  where $C$ is a constant with $C = (14\,\ln(2))^{-1}$.
\end{lemma}

\subsection{Towards a probabilistic quantum Second~Law}
\label{sec:aprobabilisticsecondlaw}
The Second~Law of Thermodynamics is probably one of the most mysterious postulates ever made to justify a physical theory.
There are many different versions of it, beside others there are versions due to Clausius, Kelvin, Planck and Boltzmann, and it is not easy to see how exactly they are related or whether they are equivalent.
One might even be tempted to say that there is not such a thing as \emph{the} Second~Law of Thermodynamics.
Basically the only obvious feature shared by all these \emph{Second~Laws} is that they introduce some sort of irreversibility.
Irreversibility is an obvious property of many processes in our everyday world.
But exactly this irreversibility is in conflict with the time reversal invariance of all microscopic theories \cite{FeynmanV01,UffinkFinal,boltzmannstanford.edu}.

Trying to \emph{derive} the Second~Law from Quantum Mechanics thus seems to be a hopeless endeavor.
However, if we take (i) a tendency to equilibrate and (ii) a tendency to increase disorder as measured by some entropic quantity as the most important aspects of the Second~Law, then theorem~\ref{theorem:highaveragedeffectivedimensionisgeneric} and \ref{theorem:distancefromtimeaverage} together with theorem~\ref{theorem:initialstateindependence} and lemma~\ref{lemma:entangledeigenstasaregeneric} are sufficient to derive a probabilistic pseudo Second~Law from just standard Quantum Mechanics:
\begin{theorem}
  \label{theorem:secondlaw}
  {\bf (Probabilistic pseudo quantum Second~Law)} 
  Given an arbitrary fixed pure initial state of a large bipartite quantum system $\hiH = \hiH_S \otimes \hiH_B$ with $d_S \ll d_B$.
  The time evolution of almost every random Hamiltonian is such that the reduced state on $S$ is close to an equilibrium state for almost all times. 
  This equilibrium state does not depend on the initial state of the system and maximizes the local von~Neumann entropy.
\end{theorem}
\begin{proof}
  Theorem~\ref{theorem:initialstateindependence} proofs initial state independence if the eigenstates of the Hamiltonian are highly entangled and lemma~\ref{lemma:entangledeigenstasaregeneric} shows that this is the case for almost all random Hamiltonians if $d_S \ll d_B$.
  The randomness of the initial state needed in theorem~\ref{theorem:highaveragedeffectivedimensionisgeneric} can be absorbed in the randomness of the eigenbasis of the Hamiltonian.
  A fixed initial state has a high effective dimension with respect to almost every Hamiltonian.
  Under this condition theorem~\ref{theorem:distancefromtimeaverage} ensures equilibration of all subsystems with $d_S \ll d_B$.
\end{proof}

First of all it is crucial to note that the theorem stated above is a \emph{statistical} assertion.
It does not deny the possibility that a system near the completely mixed, equilibrium state suddenly becomes purer and therefore is immune to both the reversibility paradox raised by Loschmidt and the recurrence objection raised by Poincar\'{e} against Boltzmann's famous H-Theorem \cite{UffinkFinal,PhysRev.107.33} (Figure~\ref{fig:recurenceandtimereversalobjection}).
Theorem~\ref{theorem:secondlaw} states that a system started in a pure state will have a tendency to evolve towards less pure states, thereby increasing the systems von~Neumann entropy.
This is what Ehrenfest would have called a statistical H-Theorem \cite{UffinkFinal,RevModPhys.27.289}.
Such a statistical assertion is the strongest that is compatible with a time reversal invariant microscopic theory and therefore the best we can hope for in the framework of Quantum Mechanics (again see \cite{UffinkFinal}, especially chapter 4). 

\begin{figure}[tb]
  \centering\footnotesize
  \subfloat[]{%
    \inputTikZ{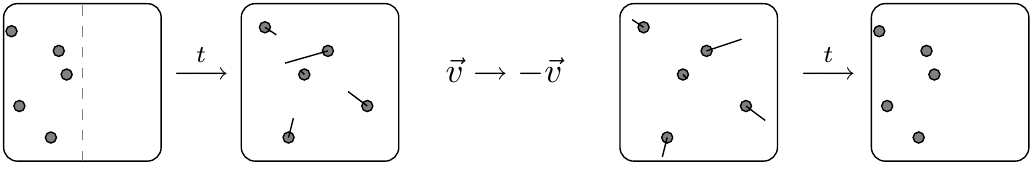}
  } \\ \subfloat[]{%
    \inputTikZ{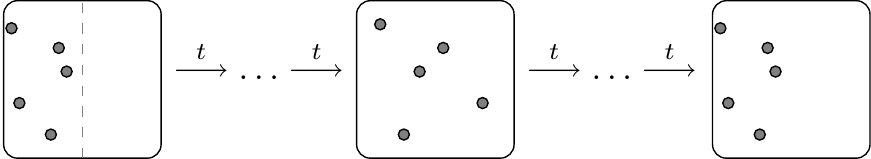}
  }
  \caption{Schematic particles-in-a-box pictures depicting Loschmidt's time reversibility paradox and Poincar\'{e}'s recurrence objection, which show that conventional Second Law of thermodynamics and in particular Boltzmann's H-Theorem contradict time reversal microscopic dynamics\cite{UffinkFinal,PhysRev.107.33}. The former (a) is an argument which shows that for every initial non-equilibrium situation of a time reversal invariant system that evolves towards equilibrium there is an initial equilibrium situation, that can be constructed from the former by inverting the time evolution (which is equivalent to reverting all momenta in the classical setting), that evolves out of equilibrium. The later (b) is a theorem which shows that certain energy conserving systems recur arbitrarily close to their initial state after a possibly very long time.}   
  \label{fig:recurenceandtimereversalobjection}
\end{figure}

As said above there exist various versions of the Second~Law, but they all apply to \emph{thermally isolated} or \emph{closed} systems, while our theorem~\ref{theorem:secondlaw} is a statement about small subsystems of large \emph{fully interactive} quantum system.
Although these canonical versions of the Second~Law obviously contradict standard Quantum Mechanics a legitimate question might be: What is the connection between theorem~\ref{theorem:secondlaw} and the canonical versions of the Second~Law?
According to the interpretation of the system-bath setup this question can be answered in two different ways:

First, one might argue that it is generally impossible to isolate a macroscopic system from its environment and think of the bath as the laboratory and the subsystem as the thermodynamic system under consideration.
Of course the Hamiltonian that describes such a situation is most likely not of the random form for which our theorem holds.

Second, taking into account that a realistic measurement on a macroscopic system, such a measurement of the pressure in a gas container or the magnetization of a macroscopic magnet, usually act only on a small part of the whole system (see for example \cite{RevModPhys.27.289} especially  p. 306 and \cite{PhysRev.108.17}).
In the first example only the average momentum of the atoms that hit the detector membrane during the observation time is measured.
In the second example the observable is a sum of operators that measure the magnetization of each individual magnetic moment in the magnet and thus a sum of local observables acting on \emph{reduced subsystem states} of a large interacting system.
One may then think of the joint system as the system under consideration (the gas in the container or the magnet) and may assume perfect isolation from the environment.\footnote{Although the very fact that the gas is trapped in the contained implies that there \emph{must} be some form of interaction.}
This leads to a highly speculative point of view, namely that our impression that large isolated systems tend to equilibrium might just be an illusion that arises from the fact that when we think that we measure properties of macroscopic systems in reality only a relatively small subsystem is measured. 
From finding such subsystems in an equilibrium state we spuriously infer that the whole system must be in an equilibrium state too.
This conclusion would be correct in a classical world, but, as can be seen from the discussion above, this is not necessarily correct if the system is quantum.

Theorem~\ref{theorem:secondlaw} is somewhat stronger than what one would have wanted to show.
It implies for example that by measuring the subsystem we can get \emph{no information about the total energy of the composite system at all}.
The reason for this is simply that realistic Hamiltonians usually comprise only short range interactions such that energy is an extensive quantity and that the decomposition of the joint system into a bath and a system usually corresponds to a division of the whole system in two spatially disjoint regions. 
Mathematically typical Hamiltonians are not necessarily realistic Hamiltonians.

It is therefore of outstanding importance not to misinterpret the above result.
By making the above statement the author does not want to imply that random Hamiltonians are in any way realistic.
The author is well aware that this is not the case.
However theorem~\ref{theorem:secondlaw} shows that there is a \emph{natural tendency} to approach equilibrium and to maximize entropy.
Traces of this tendency are expected to be found also in more realistic situations.
The result raises the hope that by imposing further constraints on the Hamiltonian, like finite interaction range, extensivity of energy, or conservation of certain quantities, one might be able to proof a theorem that comes closer to a realistic Second~Law of Thermodynamics than the one presented above.

\chapter{Conclusions}
\label{sec:conclusions}
We have made an attempt to rebuild the foundations of Statistical Mechanics and Thermodynamics form an underlying microscopic theory, namely Quantum Mechanics.
Instead of relying on additional postulates we seek for a justification of the methods of Statistical Mechanics from first principles.
Our approach is genuine quantum as randomness and statistical behavior emerge as a consequence of uncertainty relations and entanglement with the environment.

The approach gives a measure theoretic justification for the microcanonical and canonical ensemble and is capable of explaining the tendency to evolve towards equilibrium in a natural way.
New bounds on the time scales on which equilibration happens have been obtained.
We have identified a generic decoherence mechanism that makes the states of systems that interact weakly with an environment become approximately diagonal in the energy eigenbasis and we have derived a Second Law of Thermodynamics from Quantum Mechanics.
In addition, the measure theoretic foundations of the approach are strengthened by giving new bounds on the average effective dimension in the mean energy ensemble and for initial product states.

The author would like to thank Andreas Winter for introducing him to the field, the ongoing support and the valuable discussions, Haye Hinrichsen and Peter Janotta for the great amount of time, the numerous discussions and the helpful comments concerning this manuscript, Jens Eisert and Markus M\"{u}ller for the inspiriting discussions and comments on this work, as well as Cedric Beny, Myungshik Kim and Jaeyoon Cho for the constructive criticism.
The author is grateful for being supported by the German National Academic Foundation.

%%%% Appendix%%%%%%%%%%%%%%%%%%%%%%%%%%%%%%%%%%%%%%%%%%%
\appendix
\chapter{Distance measures for quantum states}
\label{appendix:distancemeasuresforquantumstates}
In this work we make use of a couple of different norms and distance measures for quantum states.
The most important of which is the trace norm and the trace distance.

Let $\rho \in \mathcal{M}(\hiH)$ be a normalized density matrix with eigenvalues $\{p_k\}$ and $d = \dim(\hiH)$.
Its $L_1$-norm, or trace norm, of $\rho$ is defined to be
\begin{equation}
  \|\rho\|_1 = \Tr|\rho| = \Tr[\sqrt{\rho^\dagger\,\rho}] = \sum_k |p_k| .
\end{equation}
The trace distance is proportional to the metric induced by this norm
\begin{align}
  \tracedistance(\rho,\sigma) &= \frac{1}{2} \|\rho -\sigma \|_1= \frac{1}{2} \Tr |\rho - \sigma| \\
  \label{eq:operatormaximasationdefinitionofthetracenorm}
  &= \max_{0 \leq A \leq \mathds{1}} \Tr[A (\rho - \sigma)] \\
  &= \max_{\Pi  \in \mathcal{P}(\hiH)} \Tr[\Pi  (\rho - \sigma)] .
\end{align}
That the definitions given above are indeed equivalent can be seen as follows: $(\rho - \sigma)$ is a traceless hermitian operator, the normalized hermitian observable $A$ that maximizes \eqref{eq:operatormaximasationdefinitionofthetracenorm} thus is the projector onto the positive subspace of $(\rho - \sigma)$, which is in turn equal to $1/2 \Tr |\rho - \sigma|$ \cite{nielsenm.a.c}. 

Equation \eqref{eq:operatormaximasationdefinitionofthetracenorm} shows what makes the trace distance $\tracedistance(\rho,\sigma)$ so special among other possible distance measures:
It can be interpreted as the physical distinguishably of $\rho$ and $\sigma$.
If two states are close to each another with respect to trace distance there is no measurement by which they can be distinguished.

As the trace norm and the trace distance depend only on the eigenvalues of their arguments it is manifest that both are invariant under unitary transformations 
\begin{equation}
  \forall U \in SU(d):\quad \tracedistance(U\,\rho\,U^\dagger,U\,\sigma\,U^\dagger) = \tracedistance(\rho,\sigma) .
\end{equation}
In fact, the trace norm is the largest unitary invariant norm in the sense that $|||A||| \leq \|A\|_1$ for all hermitian operators $A$ and all unitary invariant norms $|||\cdot|||$ \cite{bhatia}.

If two density matrices $\rho$ and $\sigma$ with eigenvalues $\{p_k\}$ and $\{q_k\}$ commute, the trace distance reduces to one half of the classical $L_1$-distance of their spectra \cite{nielsenm.a.c}
\begin{equation}
  [\rho,\sigma] = 0 \quad\Longrightarrow\quad \tracedistance(\rho,\sigma) = \frac{1}{2} \sum_k |p_k - q_k|
\end{equation}
and for pure states the trace distance is related to the Hilbert-Schmidt norm via 
\begin{equation}
  \tracedistance(\psi,\varphi) = \sqrt{1-|\braket{\psi}{\varphi}|^2} \leq \|\ket{\psi} - \ket{\varphi}\|_2 ,
\end{equation}
while for mixed states one has \cite{Linden09}
\begin{equation}
  \tracedistance(\psi,\varphi) \leq \frac{1}{2} \sqrt{d\,\Tr[(\rho-\sigma)^2]} .
\end{equation}

\chapter{The Haar Measure}
\label{appendix:thehaarmeasure}
A Haar measure is a Borel measure $\mu$ in a locally compact topological group $X$, such that $\mu(x\,B) = \mu(B)$ for every $x \in X$ and every Borel set $B \subseteq X$  and $\mu(B) > 0$ for every non-empty Borel open set $B$.
Equivalently one can say that a Haar measure is a left invariant Borel measure that is non-vanishing \cite{halmos}.

Of particular importance in the context of Quantum Mechanics is the special unitary group in $d$ dimensions $SU(d)$.
It can be shown that on $SU(d)$ there exists a left and right invariant measure $\mu_{SU(d)}$ which is unique up to normalization \cite{halmos}.
The normalization can be chosen such that $\mu_{SU(d)}(SU(d)) = 1$. 
This measure is what we call \emph{the} Haar measure on the special unitary group.

The measure $\mu_{SU(d)}$ induces a ``uniform'' measure on the set of pure quantum state vectors of a $d$ dimensional quantum system in the following way:
Choose an arbitrary fixed pure reference state and apply random unitary transformations from $\mu_{SU(n)}$ to it.
Due to the left invariance of $\mu_{SU(n)}$ the resulting distribution of pure states will also be invariant under unitary transformations and therefore is ``uniform'' on the set of pure states.

For actually sampling random pure states the above construction is not very useful as sampling random unitary transformations from $\mu_{SU(n)}$ is quite expensive.
Fortunately there is a much easier method to construct ``uniform'' random pure states of a $d$ dimensional quantum system:
Choose the entries of a $2\,d$ dimensional real vector $\vec{x}$ according to the normal distribution.
Normalize $\vec{x}$ such that $\|\vec{x}\| = 1$ in the standard euclidean norm $\|\cdot\|$.
Use the $2\,d$ entries of $\vec{x}$ as the real and imaginary parts of the $d$ complex expansion coefficients of the random vector with respect to some arbitrary fixed orthonormal basis.

\chapter{Levy's lemma and its application in Quantum Mechanics}
\label{appendix:levyslemma}

Levy's Lemma is a measure concentration result useful in high dimensional geometry.
It has previously been used in the context of Quantum Mechanics in \cite{Popescu05} and \cite{Linden09}.
\begin{lemma}
  \label{lemma:levyslemma}
  {\bf Levy's lemma \cite{milman2001}}
  Let $f: S^{d-1} \rightarrow \mathbbm{R}$ be a real valued function on the $(d-1)$-dimensional Euclidean hypersphere with Lipschitz constant
  \begin{equation}
    \eta = \sup_{\vec{x}_1,\vec{x}_2} \frac{|f(\vec{x}_1) - f(\vec{x}_2)|}{\|\vec{x}_1 - \vec{x}_2\|}
  \end{equation}
  where $\|\cdot\|$ denotes the standard euclidean norm.
  Then, for a uniformly random point $\vec{x} \in S^{d-1}$,
  \begin{equation}
    \probability\left\{|f(\vec{x}) - \expect{f}| \geq \epsilon \right\} \leq 2 \ee^{-\frac{C\,d\,\epsilon^2}{\eta^2}}
  \end{equation}
  where $\expect{f}$ is the expectation value of $f$ for uniformly random vectors from $S^{d-1}$ and $C$ is a constant with $C=(9\,\pi^3)^{-1}$.
\end{lemma}
To see how this purely mathematical lemma can be applied to Quantum Mechanics note that normalized pure states chosen from some subspace $\hiH_R$ with dimension $d_R$ can be thought of as lying on a $(2\,d_R - 1)$-dimensional hypersphere with coordinates given by the real and imaginary components of the expansion coefficients with respect to some orthonormal basis $\{\ket{i}\}$ of $\hiH_R$
\begin{align}
  x_{2i}(\psi) &= \Re[\braket{i}{\psi}] & x_{2i+1}(\psi) &= \Im[\braket{i}{\psi}]
\end{align}
and that in this coordinate system the standard Euclidean metric on $\mathbbm{R}^{2\,d_R}$ and the metric induced by the standard Hilbert space norm coincide in the sense that
\begin{equation}
  \|\vec{x}_1 - \vec{x}_2\| = \| \ket{\psi_1} - \ket{\psi_2} \|_2 .
\end{equation}

%%%% Bibliography%%%%%%%%%%%%%%%%%%%%%%%%%%%%%%%%%%%%%%%

\providecommand{\href}[2]{#2}\begingroup\raggedright\endgroup


\begin{thebibliography}{10}

\bibitem{UffinkFinal}
J.~Uffink, ``Compendium of the foundations of classical statistical physics,''.
  \url{http://philsci-archive.pitt.edu/archive/00002691/}.

\bibitem{RevModPhys.27.289}
D.~T. Haar, ``Foundations of statistical mechanics,'' {\em Rev. Mod. Phys.}
  (1955) no.~27, 289 -- 338.

\bibitem{FeynmanV01}
R.~P. Feynman, {\em Mechanics, Radiation, and Heat}, vol.~1 of {\em Lectures on
  Physics}.
\newblock Addison Wesley Longman, 1970.

\bibitem{Bassi03}
A.~Bassi and G.~Ghirardi, ``Dynamical reduction models,''
  \href{http://dx.doi.org/10.1016/S0370-1573(03)00103-0}{{\em Physics Reports}
  {\bf 379} (2003) no.~5-6, 257}.

\bibitem{Breuer02}
H.-P. Breuer and F.~Petruccione, {\em The Theory of Open Quantum Systems}.
\newblock Oxford University Press, 2002.

\bibitem{RevModPhys.75.715}
W.~H. Zurek, ``Decoherence, einselection, and the quantum origins of the
  classical,'' {\em Rev. Mod. Phys.} (2003) no.~75, 715--775.

\bibitem{zeh96}
E.~Joos, H.~Zeh, C.~Kiefer, D.~Giulini, J.~Kupsch, and I.-O. Stamatescu, {\em
  Decoherence and the Appearance of a Classical World in Quantum Theory}.
\newblock Springer, 1996.

\bibitem{Gemmer09}
J.~Gemmer, M.~Michel, and G.~Mahler,
  \href{http://dx.doi.org/10.1007/978-3-540-70510-9}{{\em Quantum
  Thermodynamics}}, vol.~784.
\newblock Springer, Berlin / Heidelberg, 2009.

\bibitem{Hornberger09}
K.~Hornberger, ``Introduction to decoherence theory,''
  \href{http://dx.doi.org/10.1007/978-3-540-88169-8_5}{{\em Lect. Notes Phys.}
  {\bf 768} (2009)  223--278},
  \href{http://arxiv.org/abs/quant-ph/0612118v3}{{\tt quant-ph/0612118v3}}.

\bibitem{tasaki98}
H.~Tasaki, ``Quantum dynamics to the canonical distribution: General picture
  and a rigorous example,'' {\em Phys.Ref.Lett} {\bf 80} (1998) no.~7,
  1373--1376.

\bibitem{vonneumann1929}
J.~Von~Neumann, ``Beweis des ergodensatzes und des h-theorems in der neuen
  mechanik,'' {\em Zeitschrift f{\"{u}}r Physik A} {\bf 57} (1929) no.~1-2,
  30--70, \href{http://arxiv.org/abs/1003.2133v1}{{\tt 1003.2133v1}}.

\bibitem{Schroedinger27}
E.~Schr\"{o}dinger, ``Energieaustausch nach der wellenmechanik,''
  \href{http://dx.doi.org/10.1002/andp.19273881504}{{\em Annalen der Physik}
  {\bf 388} (1927) no.~15, 956}.

\bibitem{slloydthesis}
S.~Lloyd, {\em Black Holes, Demons and the Loss of Coherence: How complex
  systems get information, and what they do with it}.
\newblock PhD thesis, Rockefeller University, April, 1991.

\bibitem{Gemmer02}
J.~Gemmer and G.~Mahler, ``Distribution of local entropy in the hilbert space
  of bi-partite quantum systems: Origin of jaynes{'} principle,''
  \href{http://arxiv.org/abs/quant-ph/0201136v1}{{\tt quant-ph/0201136v1}}.

\bibitem{Popescu05}
S.~Popescu, A.~J. Short, and A.~Winter, ``The foundations of statistical
  mechanics from entanglement: Individual states vs. averages,''
  \href{http://arxiv.org/abs/quant-ph/0511225v3}{{\tt quant-ph/0511225v3}}.
  \url{http://www.citebase.org/abstract?id=oai:arXiv.org:quant-ph/0511225}.

\bibitem{Popescu06}
S.~Popescu, A.~J. Short, and A.~Winter, ``Entanglement and the foundations of
  statistical mechanics,'' \href{http://dx.doi.org/10.1038/nphys444}{{\em
  Nature Physics} {\bf 2} (2006) no.~11, 754}.

\bibitem{Goldstein06}
S.~Goldstein, ``Canonical typicality,''
  \href{http://dx.doi.org/10.1103/PhysRevLett.96.050403}{{\em Physical Review
  Letters} {\bf 96} (2006) no.~5, 050403}.

\bibitem{Cho09}
J.~Cho and M.~S. Kim, ``Emergence of canonical ensembles from pure quantum
  states,'' \href{http://arxiv.org/abs/0911.2110v1}{{\tt 0911.2110v1}}.
  \url{http://www.citebase.org/abstract?id=oai:arXiv.org:0911.2110}.

\bibitem{ledoux01}
M.~Ledoux, {\em The Concentration of Measure Phenomenon}, vol.~89 of {\em
  Mathematical Surveys and Monographs}.
\newblock Americal Mathematical Society, 2001.

\bibitem{PhysRev.106.62}
E.~Jaynes, ``Information theory and statistical mechanics,''
  \href{http://dx.doi.org/10.1103/PhysRev.106.620}{{\em Physical Review} {\bf
  106} (1957) no.~4, 620}.

\bibitem{PhysRev.108.17}
E.~Jaynes, ``Information theory and statistical mechanics. ii,''
  \href{http://dx.doi.org/10.1103/PhysRev.108.171}{{\em Physical Review} {\bf
  108} (1957) no.~2, 171}.

\bibitem{PhysRevE.50.88}
M.~Srednicki, ``Chaos and quantum thermalization,'' {\em Phys. Rev. E} {\bf 50}
  (1994) no.~888, .

\bibitem{Gemmer06}
J.~Gemmer and M.~Michel, ``Thermalization of quantum systems by finite baths,''
  \href{http://dx.doi.org/10.1209/epl/i2005-10363-0}{{\em Europhysics Letters}
  {\bf 73} (2006)  1}, \href{http://arxiv.org/abs/quant-ph/0511023v1}{{\tt
  quant-ph/0511023v1}}.

\bibitem{Reimann08}
P.~Reimann, ``Foundation of statistical mechanics under experimentally
  realistic conditions,''
  \href{http://dx.doi.org/10.1103/PhysRevLett.101.190403}{{\em Phys. Rev.
  Lett.} {\bf 101} (2008) no.~19, 190403}.

\bibitem{Linden09}
N.~Linden, S.~Popescu, A.~J. Short, and A.~Winter, ``Quantum mechanical
  evolution towards thermal equilibrium,''
  \href{http://dx.doi.org/10.1103/PhysRevE.79.061103}{{\em Physical Review E}
  {\bf 79} (2009) no.~6, 061103}, \href{http://arxiv.org/abs/0812.2385v1}{{\tt
  0812.2385v1}}.

\bibitem{Bartsch09}
C.~Bartsch and J.~Gemmer, ``Dynamical typicality of quantum expectation
  values,'' \href{http://dx.doi.org/10.1103/PhysRevLett.102.110403}{{\em
  Physical Review Letters} {\bf 102} (2009) no.~11, 110403}.

\bibitem{Borowski03}
P.~Borowski, J.~Gemmer, and G.~Mahler, ``Relaxation into equilibrium under pure
  schr{\"{o}}dinger dynamics,'' {\em The European Physical Journal B} {\bf 35}
  (2003)  255, \href{http://arxiv.org/abs/quant-ph/0310176v1}{{\tt
  quant-ph/0310176v1}}.
  \url{http://www.citebase.org/abstract?id=oai:arXiv.org:quant-ph/0310176}.

\bibitem{Wang08}
W.-g. Wang, J.~Gong, G.~Casati, and B.~Li, ``Entanglement-induced decoherence
  and energy eigenstates,''
  \href{http://dx.doi.org/10.1103/PhysRevA.77.012108}{{\em Physical Review A}
  {\bf 77} (2008) no.~1, 012108}.

\bibitem{Devi09}
A.~R.~U. Devi and A.~K. Rajagopal, ``Dynamical evolution of quantum oscillators
  toward equilibrium,''
  \href{http://dx.doi.org/10.1103/PhysRevE.80.011136}{{\em Phys. Rev. E} {\bf
  80} (2009) no.~011136, }.

\bibitem{Cramer08}
M.~Cramer, C.~M. Dawson, J.~Eisert, and T.~J. Osborne, ``Exact relaxation in a
  class of non-equilibrium quantum lattice systems,'' {\em Physical Review
  Letters} {\bf 100} (2008)  030602,
  \href{http://arxiv.org/abs/cond-mat/0703314v2}{{\tt cond-mat/0703314v2}}.
  \url{http://www.citebase.org/abstract?id=oai:arXiv.org:cond-mat/0703314}.

\bibitem{Cramer09}
M.~Cramer and J.~Eisert, ``A quantum central limit theorem for non-equilibrium
  systems: Exact local relaxation of correlated states,''
  \href{http://arxiv.org/abs/0911.2475v1}{{\tt 0911.2475v1}}.

\bibitem{Merkli09}
M.~Merkli, G.~P. Berman, and I.~M. Sigal, ``Resonant perturbation theory of
  decoherence and relaxation of quantum bits,''
  \href{http://arxiv.org/abs/0911.3122v1}{{\tt 0911.3122v1}}.
  \url{http://www.citebase.org/abstract?id=oai:arXiv.org:0911.3122}.

\bibitem{Schroeder10}
H.~Schr{\"{o}}der and G.~Mahler, ``Work exchange between quantum systems: The
  spin-oscillator model,''
  \href{http://dx.doi.org/10.1103/PhysRevE.81.021118}{{\em Physical Review E}
  {\bf 81} (2010) no.~2, 021118}.

\bibitem{Linden09-2}
N.~Linden, S.~Popescu, and P.~Skrzypczyk, ``How small can thermal machines be?:
  Towards the smallest possible refrigerator,''
  \href{http://arxiv.org/abs/0908.2076v1}{{\tt 0908.2076v1}}.
  \url{http://www.citebase.org/abstract?id=oai:arXiv.org:0908.2076}.

\bibitem{0907.1267v1}
N.~Linden, S.~Popescu, A.~J. Short, and A.~Winter, ``On the speed of
  fluctuations around thermodynamic equilibrium,''
  \href{http://arxiv.org/abs/0907.1267v1}{{\tt 0907.1267v1}}.

\bibitem{0907.0108v1}
S.~Goldstein, J.~L. Lebowitz, C.~Mastrodonato, R.~Tumulka, and N.~Zanghi,
  ``Normal typicality and von neumann's quantum ergodic theorem,''
  \href{http://arxiv.org/abs/0907.0108v1}{{\tt 0907.0108v1}}.

\bibitem{PhysRev.114.94}
P.~Bocchieri and A.~Loinger, ``Ergodic foundation of quantum statistical
  mechanics,'' {\em Phys.Rev.} {\bf 114} (1959) no.~4, 948--951.

\bibitem{Farquhar57}
I.~E. Farquhar and P.~T. Landsberg, ``On the quantum-statistical ergodic and
  h-theorems,'' {\em Proc. Royal Soc. London A} (1957) no.~239, 134--144.

\bibitem{Gemmer01}
J.~Gemmer, ``Quantum approach to a derivation of the second law of
  thermodynamics,'' \href{http://dx.doi.org/10.1103/PhysRevLett.86.1927}{{\em
  Physical Review Letters} {\bf 86} (2001) no.~10, 1927}.

\bibitem{Brandao07}
F.~G. S.~L. Brandao and M.~B. Plenio, ``A reversible theory of entanglement and
  its relation to the second law,''
  \href{http://arxiv.org/abs/0710.5827v2}{{\tt 0710.5827v2}}.
  \url{http://www.citebase.org/abstract?id=oai:arXiv.org:0710.5827}.

\bibitem{Polkovnikov08}
A.~Polkovnikov, ``Microscopic diagonal entropy and its connection to basic
  thermodynamic relations,'' \href{http://arxiv.org/abs/0806.2862v7}{{\tt
  0806.2862v7}}.
  \url{http://www.citebase.org/abstract?id=oai:arXiv.org:0806.2862}.

\bibitem{0708.1324v1}
M.~Rigol, V.~Dunjko, and M.~Olshanii, ``Thermalization and its mechanism for
  generic isolated quantum systems.,''
  \href{http://dx.doi.org/10.1038/nature06838}{{\em Nature} {\bf 452} (2008)
  no.~7189, 854--8}, \href{http://arxiv.org/abs/0708.1324v1}{{\tt
  0708.1324v1}}.

\bibitem{0904.1501v1}
S.~Yuan, M.~I. Katsnelson, and H.~De~Raedt, ``Origin of the canonical ensemble:
  Thermalization with decoherence,''
  \href{http://arxiv.org/abs/0904.1501v1}{{\tt 0904.1501v1}}.

\bibitem{Wu09}
C.~Wu and H.~Guo, ``Non-markovian dynamics without using quantum trajectory,''
  \href{http://arxiv.org/abs/0912.0771v1}{{\tt 0912.0771v1}}.
  \url{http://www.citebase.org/abstract?id=oai:arXiv.org:0912.0771}.

\bibitem{PhysRevA.43.20}
J.~M. Deutsch, ``Quantum statistical mechanics in closed systems,'' {\em
  Phys.Rev. A} {\bf 43} (1991) no.~4, 2046--2049.

\bibitem{duistermaat99}
J.~J. Duistermaat and J.~A.~C. Polk, {\em Lie Groups}.
\newblock Springer-Verlag, Berlin, 1999.

\bibitem{Schwabl02}
F.~Schwabl, {\em Statistical mechanics}.
\newblock Springer, Berlin, 2002.

\bibitem{logikderforschung}
K.~R. Popper, {\em Logik der Forschung}.
\newblock Mohr Siebeck, T{\"{u}}bingen, 11~ed., 2005.

\bibitem{Bender05}
C.~M. Bender, D.~C. Brody, and D.~W. Hook, ``Solvable model of quantum
  microcanonical states,''
  \href{http://dx.doi.org/10.1088/0305-4470/38/38/L01}{{\em Journal of Physics
  A: Mathematical and General} {\bf 38} (2005) no.~38, L607}.

\bibitem{Brody05}
D.~C. Brody, D.~W. Hook, and L.~P. Hughston, ``Microcanonical distributions for
  quantum systems,'' \href{http://arxiv.org/abs/quant-ph/0506163v1}{{\tt
  quant-ph/0506163v1}}.
  \url{http://www.citebase.org/abstract?id=oai:arXiv.org:quant-ph/0506163}.

\bibitem{Mueller09}
M.~M\"{u}ller, D.~Gross, and J.~Eisert, ``Concentration of measure and the mean
  energy ensemble.'' unpublished, 2009.

\bibitem{greiner}
S.~Greiner, Neise, {\em Thermodynamics and Statistiacal Mechanics}.
\newblock Springer Verlag, 1995.

\bibitem{noltingstatistischephysik01}
W.~Nolting, {\em Statistische Physik}, vol.~6 of {\em Grundkurs Theoretische
  Physik}.
\newblock Springer Verlag, 5~ed., 2005.

\bibitem{Brody07}
D.~C. Brody, D.~W. Hook, and L.~P. Hughston, ``Quantum phase transitions
  without thermodynamic limits,''
  \href{http://dx.doi.org/10.1098/rspa.2007.1865}{{\em Proceedings of the Royal
  Society A: Mathematical, Physical and Engineering Sciences} {\bf 463} (2007)
  no.~2084, 2021}.

\bibitem{0909.3175v1}
B.~Fresch and G.~J. Moro, ``Typicality in ensembles of quantum states: Monte
  carlo sampling vs analytical approximations,''
  \href{http://arxiv.org/abs/0909.3175}{{\tt 0909.3175}}.

\bibitem{1003.4982}
M.~Mueller, D.~Gross, and J.~Eisert, ``Concentration of measure for quantum
  states with a fixed expectation value,''
  \href{http://arxiv.org/abs/1003.4982v1}{{\tt 1003.4982v1}}.

\bibitem{sd268119982}
O.~Krafft, ``An arithmetic{---}harmonic-mean inequality for nonnegative
  definite matrices,''
  \href{http://dx.doi.org/10.1016/S0024-3795(97)00046-3}{{\em Linear Algebra
  and its Applications} {\bf 268} (1998)  243}.

\bibitem{PhysRev.107.33}
P.~Bocchieri and A.~Loinger, ``Quantum recurrence theorem,'' {\em Phys. Rev.}
  {\bf 107} no.~2, 337{--}338.

\bibitem{Barthel08}
T.~Barthel and U.~Schollw{\"{o}}ck, ``Dephasing and the steady state in quantum
  many-particle systems,''
  \href{http://dx.doi.org/10.1103/PhysRevLett.100.100601}{{\em Physical Review
  Letters} {\bf 100} (2008) no.~10, 100601}.

\bibitem{Kimura07}
G.~Kimura, H.~Ohno, and H.~Hayashi, ``How to detect a possible correlation from
  the information of a subsystem in quantum-mechanical systems,''
  \href{http://dx.doi.org/10.1103/PhysRevA.76.042123}{{\em Physical Review A}
  {\bf 76} (2007) no.~4, 042123}.

\bibitem{Kimura09-1}
G.~Kimura, H.~Ohno, and M.~Mosonyi, ``Relation between the dynamics of the
  reduced purity and correlations,''
  \href{http://arxiv.org/abs/0910.5297v1}{{\tt 0910.5297v1}}.
  \url{http://www.citebase.org/abstract?id=oai:arXiv.org:0910.5297}.

\bibitem{nielsenm.a.c}
M.~A. Nielsen and I.~L. Chuang, {\em Quantum Computation and Quantum
  Information}.
\newblock Cambridge University Press, 2007.

\bibitem{PhysRevD.26.18}
W.~H. Zurek, ``Environment-induced superselection rules,'' {\em Phys. Rev. D}
  {\bf 26} (1982) no.~8, 1862--1880.

\bibitem{Gogolin09-1}
C.~Gogolin, ``Einselection without pointer states,''
  \href{http://arxiv.org/abs/0908.2921v2}{{\tt 0908.2921v2}}.

\bibitem{PhysRevLett.82}
J.~P. Paz and W.~H. Zurek, ``Quantum limit of decoherence: Environment induced
  superselection of energy eigenstates,'' {\em Phys. Rev. Lett.} {\bf 82}
  (1999)  5181--5185.

\bibitem{Hayden06}
P.~Hayden, D.~W. Leung, and A.~Winter, ``Aspects of generic entanglement,''
  \href{http://dx.doi.org/10.1007/s00220-006-1535-6}{{\em Communications in
  Mathematical Physics} {\bf 265} (2006)  95},
  \href{http://arxiv.org/abs/quant-ph/0407049v2}{{\tt quant-ph/0407049v2}}.

\bibitem{Goldstein09}
S.~Goldstein, J.~L. Lebowitz, C.~Mastrodonato, R.~Tumulka, and N.~Zanghi, ``On
  the approach to thermal equilibrium of macroscopic quantum systems,''
  \href{http://arxiv.org/abs/0911.1724v1}{{\tt 0911.1724v1}}.
  \url{http://www.citebase.org/abstract?id=oai:arXiv.org:0911.1724}.

\bibitem{boltzmannstanford.edu}
J.~Uffink, ``Boltzmann's work in statistical physics.''
  Http://plato.stanford.edu/entries/statphys-boltzmann/.
\newblock \url{http://plato.stanford.edu/entries/statphys-Boltzmann/}.

\bibitem{bhatia}
R.~Bhatia, {\em Matrix Analysis}.
\newblock Springer Verlag, New York, 1997.

\bibitem{halmos}
P.~R. Halmos, {\em Measure Theory}.
\newblock Springer-Verlag, New York, 1974.

\bibitem{milman2001}
V.~Milman and G.~Schechtman, {\em Asymptotic Theory of Finite Dimensional
  Normed Spaces}.
\newblock Springer Verlag, LNM 1200, Berlin, 2001.

\end{thebibliography}
\end{document}